\documentclass[journal]{IEEEtran}
\ifCLASSINFOpdf
\else
\fi
\usepackage{cite}
\usepackage{bm}
\usepackage{amsmath,amsthm}
\usepackage{extarrows}
\usepackage{amssymb}
\usepackage{graphicx}
\graphicspath{{figures/}}
\usepackage{xcolor}
\usepackage{subfig}
\usepackage{commath}
\DeclareMathOperator*{\argmax}{arg\,max}
\DeclareMathOperator*{\argmin}{arg\,min}

\newcommand{\mr}{\mathrm}

\newcommand{\BE}{\begin{equation}}
\newcommand{\EE}{\end{equation}}
\newcommand{\BS}{\begin{subequations}}
\newcommand{\ES}{\end{subequations}}
\newcommand{\UH}{\mathsf{H}}   

\renewcommand{\Re}{\mr{Re}}   
\renewcommand{\Im}{\mr{Im}}   

\newcommand{\Mydef}{\overset{  \scriptscriptstyle \Delta  }{=}}

\renewcommand{\bf}{\bm}
\newtheorem{theorem}{Theorem}

\newtheorem{proposition}{Proposition}
\newtheorem{assumption}{Assumption}
\newtheorem{definition}{Definition}
\newtheorem{remark}{Remark}
\newtheorem{lemma}{Lemma}
\newtheorem{corollary}{Corollary}

\newtheorem{fact}{Fact}

\newcommand{\E}{\mathbb{E}}
\newcommand{\tE}{\tilde{\mathbb{E}}}
\newcommand{\T}{\mathcal{T}}

\newcommand*\diff{\mathop{}\!\mathrm{d}}

\newcommand{\cgauss}[2]{\mathcal{CN}\left( #1,#2 \right)}

\newcommand{\explain}[2]{\overset{\text{\tiny{#1}}}{#2}}

\newcommand{\Arg}{\mathsf{Arg}}

\newcommand{\opt}{\mathsf{opt}}

\newcommand{\Unif}{\textup{Unif}}

\newcommand{\p}[2][{}]{\mathbb{P}_{#1} \intoo{#2}}
\newcommand{\e}[2][{}]{\mathbb{E}_{#1} \sbr{#2}}

\begin{document}
%
\title{Analysis of Spectral Methods for Phase Retrieval with Random Orthogonal Matrices}
%
%


\author{Rishabh Dudeja, Milad Bakhshizadeh, Junjie Ma, Arian Maleki \\
        Department of Statistics, Columbia University} 

\maketitle

\begin{abstract}
                Phase retrieval refers to algorithmic methods for recovering a signal from its phaseless measurements. There has been recent interest in understanding the performance of local search algorithms that work directly on the non-convex formulation of the problem. Due to the non-convexity of the problem, the success of these local search algorithms depends heavily on their starting points. The most widely used initialization scheme is the spectral method, in which the leading eigenvector of a data-dependent matrix is used as a starting point. Recently, the performance of the spectral initialization was characterized accurately for measurement matrices with independent and identically distributed entries. This paper aims to obtain the same level of knowledge for isotropically random column-orthogonal matrices, which are substantially better models for practical phase retrieval  systems. Towards this goal, we consider the asymptotic setting in which the number of measurements $m$, and the dimension of the signal, $n$, diverge to infinity with $m/n = \delta\in(1,\infty)$, and obtain a simple expression for the overlap between the spectral estimator and the true signal vector. 
        \end{abstract}
 

%
\IEEEpeerreviewmaketitle

	\section{Introduction}\label{sec:intro}

	Phase retrieval refers to the problem of recovering a signal $\bm x_\star \in \mathbb{C}^n$ from a set of phaseless linear observations $\bm{y} \in \mathbb{R}^m$. Under the absence of the measurement noise, the acquisition process is modeled as 
	\[ 
	{y}_i = |(\bm{A} \bm{x_{\star}})_i|,
	\]
	where $\bm{A}\in\mathbb{C}^{m\times n}$ is a measurement matrix and $(\cdot)_i$ denotes the $i^{\rm th}$ element of a vector. The phase retrieval problem is intended to model practical imaging systems where it is difficult to measure the phase of the measurements \cite{Shechtman15}.  A number of recent recovery algorithms pose Phase retrieval as a non-convex optimization problem, and employ a local search algorithm to find the minimizer \cite{CaLiSo15,ChenCandes17,Wang2016,Zhang2016reshaped}. For instance, the well known Wirtinger Flow algorithm \cite{CaLiSo15} solves the optimization problem:
	\begin{equation}\label{eq:wirtingerflow}
	\min_{\bm{x}}\quad \sum_{i=1}^m (y_i^2 - |\bm{a}_i^{\UH} \bm{x}|^2)^2,
	\end{equation}
	using gradient descent. 
	
	Since the optimization problem \eqref{eq:wirtingerflow} is non-convex,  the initialization can have an impact on the success of local search algorithms. The most widely used initialization scheme, known as spectral initialization \cite{Eetrapalli2013,ChenCandes17,Wang2016,Lu17,mondelli2017fundamental,Lu2018}, uses the leading eigenvector of the following data-dependent matrix:  
	\begin{equation}\label{eq:Mdef}
	{\bm M} \Mydef {\bm A}^{\UH} \bm{T} \bm{A}
	\end{equation}
	as the starting point for local search algorithms. In the above equation, $\bm{T}= {\rm Diag}(\T(y_1), \T(y_2), \ldots, \T(y_m))$, and $\T(\cdot)$ denotes a suitable trimming function. Let $\hat{\bm{x}}$ denote the leading eigenvector of $\bm{M}$ normalized to have unit Euclidean ($\ell_2$) norm. That is, 
	\begin{align}
	\hat{\bm x} & \Mydef \max_{\|\bm x\| = 1} \bm x^\UH \bm M \bm x.
	\label{eigen_value_problem}
	\end{align}The earliest analysis \cite{Eetrapalli2013,CaLiSo15} of the spectral estimator  showed that if number of measurements $m$ is large enough (for a fixed $n$), then the leading eigenvector of $\bm{M}$ is a consistent estimator of the true signal vector. However these analyses had two drawbacks: (i) They only provide information about the order of measurements required for a successful initialization and not a sharp requirement on the sampling ratio $m/n$, (ii) These analyses fail to capture the difference in the performance of various trimming functions. Recently, Lu and Li \cite{Lu17} have analyzed the spectral estimator for measurement matrices that are composed of independent and identically distributed (i.i.d.) standard normal entries in the high dimensional asymptotic regime. More specifically, Lu and Li considered the asymptotic setting in which $m,n \rightarrow \infty$, $m/n = \delta$, and obtained a sharp characterization for the overlap between the leading eigenvector and the true signal. In follow up work by Mondelli and Montanari \cite{mondelli2017fundamental} and Luo, Alghamdi and Lu \cite{Lu2018} this characterization was leveraged to design optimal trimming functions. For the optimal trimming function, the overlap $|\hat{\bm{x}}^\UH \bm{x}_\star|^2/\|\bm x_\star\|^2$ converges to zero when $\delta<1$, and converges to a strictly positive value otherwise. 
	
	A major assumption in the analysis of \cite{Lu17,mondelli2017fundamental,Lu2018} is that the measurement matrix $\bm{A}$ contains i.i.d. Gaussian entries. However, it is well-known that many important applications of phase retrieval are concerned with Fourier-type matrices \cite{Fienup78}. This leads to the following natural questions: (i) Are the conclusions of \cite{Lu17,mondelli2017fundamental,Lu2018} correct for other matrices that are employed in practice? (ii) Is the optimal choice of trimming that was derived in \cite{Lu17,mondelli2017fundamental,Lu2018} for Gaussian measurement matrices optimal for other  matrices employed in practice? In response to these questions, Ma \textit{et al.} \cite{ma2019spectral} considered a popular class of matrices that can be used in phase retrieval systems, known as  coded diffraction pattern (CDP) \cite{Candes15_Diffraction}. Through an extensive numerical study, the authors showed that the performance of the spectral initialization for such matrices closely approximates the performance of the spectral estimator for partial orthogonal matrices. The authors then designed an Expectation Propagation (EP) \cite{Minka2001,opper2005} algorithm for the eigenvalue problem given in (\ref{eigen_value_problem}). 
	EP algorithms had previously been proposed for partial orthogonal matrices in \cite{ma2015turbo,Ma2016} and their State Evolution (SE) had been analyzed in \cite{Rangan17,Takeuchi2017}. 
	Ma \textit{et al.} used the SE of derived EP algorithm for the eigenvalue problem to derive a (conjectured) formula for the asymptotic overlap $ |\hat{\bm{x}}^\UH \bm{x}_\star|^2/\|\bm x_\star\|^2$ between the true signal vector and the spectral initialization.
	However, while it is believed that EP algorithm indeed solves the eigenvalue problem (this has also been observed in simulations), this has not been shown rigorously. As a result of such studies, the authors conjectured that for partial orthogonal matrices if the trimming function is chosen optimally, then for $\delta> 2$,  $ |\hat{\bm{x}}^\UH \bm{x}_\star|^2/\|\bm x_\star\|^2>0$, and for $\delta< 2$, $ |\hat{\bm{x}}^\UH \bm{x}_\star|^2/\|\bm x_\star\|^2 =0$, in the asymptotic setting where $n,m=\delta n\to\infty$. As mentioned previously, the simulations in \cite{ma2019spectral} suggest that these conjectures are also likely to hold for CDP matrices.
	
	In this paper, we prove the conjectures presented in \cite{ma2019spectral} for partial orthogonal matrices using tools from the free probability theory \cite{belinschi2017outliers}. We believe this is the first theoretical justification that the expectation propagation framework can correctly predict the statistical properties of the solutions to non-convex optimization problems. The main technical step in our proof is the identification of the location of the largest eigenvalue using a subordination function \cite{belinschi2017outliers}. Interestingly, this subordination function appears naturally in the expectation propagation (EP) algorithm of \cite{ma2019spectral}.


\section{Main result}\label{main_result}
\subsection{Notation}
\subsubsection{For Linear Algebraic Aspects} For a matrix $\bm A$, $\bm A^\UH$ refers to the conjugate transpose of $\bm A$. For a matrix $\bm A \in \mathbb C^{n \times n}$, with real eigenvalues, we use ${\lambda_1(\bm A) \geq \lambda_2(\bm A) \dots \geq \lambda_n(\bm A)}$ to denote the eigenvalues arranged in descending order. We use $\sigma(\bm A)$ to refer to the spectrum of $\bm A$ which is simply the set of eigenvalues $\{\lambda_1(\bm A),\lambda_2(\bm A) \dots \lambda_n(\bm A)\}$. Finally we define the spectral measure of $\bm A$, denoted by $\mu_{\bm A}$ as, 
	\begin{align*}
	\mu_{\bm A} & \Mydef \frac{1}{n} \sum_{i=1}^n \delta_{\lambda_i(\bm A)}.
	\end{align*}
	For $m,n \in \mathbb {N}$, we denote the $m \times m$ identity matrix by $\bm I_m$ and a $m \times n$ matrix of all zero entries by $\bm 0_{m,n}$. For $m \geq n$, We also define the special matrix $\bm S_{m,n}$ as:
	\begin{align}\label{Eqn:S_def}
		\bm S_{m,n} & \Mydef \begin{bmatrix} \bm I_n \\ \bm 0_{m-n,n}  \end{bmatrix}.
	\end{align}
\subsubsection{For Complex Analytic Aspects} For a complex number $z \in \mathbb C$, $\Re(z), \Im(z), \Arg(z), |z|,\overline{z}$ refer to the real part, imaginary part, argument, modulus and conjugate of $z$. We denote the complex upper half plane and lower half planes by
\[
	\mathbb C^{+}  \Mydef \{z \in \mathbb C: \Im(z) > 0\} \;\text{and}\;
		\mathbb C^{-}  \Mydef \{z \in \mathbb C: \Im(z) < 0\}.
	\]
\subsubsection{For Probabilistic Aspects} We use $\cgauss{0}{1}$ to denote the standard, circularly symmetric, complex Gaussian distribution. $\text{Unif}(\mathbb{U}_m)$ denotes the Haar measure on the unitary group. We denote almost sure convergence, convergence in probability and convergence in distribution by $\explain{a.s.}{\rightarrow},
	\explain{P}{\rightarrow}$ and $\explain{d}{\rightarrow}$ respectively. Two random variables $X,Y$ are equal in distribution, denoted by $X \explain{d}{=} Y$ if they have the same distribution. Throughout this paper, the random variables $Z,T$ refer to the pair of random variables with the joint distribution given by $Z \sim \cgauss{0}{1}, T = \T(|Z|/\sqrt{\delta})$. For a borel probability measure $\mu$, we use $\text{Supp}({\mu})$ to denote the support of $\mu$.
\subsubsection{Miscellaneous: } Let $A$ be a subset of $\mathbb{R}$ or $\mathbb C$. $\overline{A}$ denotes the closure of $A$. The distance from a point $x \in \mathbb{R}$ to $A$ is defined by $\text{dist}(x,A) = \inf_{y \in A} |x-y|$. We define the $\epsilon$ neighborhood of $A$, denoted by $A_\epsilon$ as $$A_{\epsilon} \Mydef \{x : \text{dist}(x,A) < \epsilon \}.$$ The symbol $\emptyset$ is used to denote the empty set.

\subsection{Measurement Model and Spectral Estimator}
In the phase retrieval problem we are given $m$ observations $\bm y \in \mathbb R^m$ generated as: $$\bm y = |\bm A \bm x_\star|$$ where $\bm x_\star \in \mathbb C^n$ is the unknown signal vector and $\bm A \in \mathbb C^{m \times n}$ is the sensing matrix. We assume that $\|\bm x_\star\| = \sqrt{n}$ and that the matrix $\bm A$ is generated according to the following process: Sample $\bm H_m \in \mathbb{U}(m)$ from the Haar measure on the unitary group $\mathbb{U}(m)$ and set $\bm A$ to be the matrix formed by picking the first $n$ columns of $\bm H_m$. More formally,
\begin{align*}
\bm A = \bm H \bm S_{m,n}, \;\bm H \sim \text{Unif}(\mathbb{U}(m)), 
\end{align*}
and $\bm{S}$ is defined in \eqref{Eqn:S_def}.
An important parameter for our analysis will be \emph{the sampling ratio}, denoted by $\delta \Mydef m/n$.
Let $\mathcal{T}: \mathbb{R}_{\geq 0} \rightarrow \mathbb R$ be a trimming function. We study spectral estimators $\hat{\bm x}$  constructed as the leading eigenvector of the matrix $\bm M$, defined below:
\begin{align*}
\hat{\bm x} = \arg \max_{\|\bm u\| = 1} \bm u^\UH \bm M \bm u,
\end{align*}
where $\bm{M}= \bm A^\UH \bm T \bm A$ and $\bm T = \text{Diag}(\T(y_1),\T(y_2) \dots \T(y_m))$.

\subsection{Assumptions \& Asymptotic Framework}
We analyze the performance of the spectral estimator in an asymptotic setup where $n,m \rightarrow \infty, m/n = \delta>1$. In particular, we consider a sequence of independent phase retrieval problems realized on the same probability space with increasing $n,m$. We assume some regularity assumptions on the trimming function $\mathcal{T}$ which are stated below. 
\begin{assumption} The trimming function $\mathcal{T}$ satisfies the following conditions: 
	\begin{enumerate}
		\item $\T$ is Lipschitz continuous.
		\item $\sup_{y \geq 0} \T(y) = 1, \; \inf_{y \geq 0} \T(y) = 0$.
		\item The random variable $T$, defined by $Z \sim \cgauss{0}{1}$ and $T = \T(|Z|/\sqrt{\delta})$ has a density with respect to the Lebesgue measure on $\mathbb R$.
	\end{enumerate}
\label{regularity_assumption}
\end{assumption}

In the following remarks, we discuss why each of these assumptions are required and whether they can be relaxed. 

\begin{remark} We need the trimming function $\T$ to be Lipschitz continuous so that the trimmed measurements $\T(y_i)$ can be approximated in distribution by $\T(|Z|/\sqrt{\delta}), Z \sim \cgauss{0}{1}$. We expect  this approximation to hold under weaker smoothness hypothesis on $\T$ than Lipschitz continuity.
\end{remark}

\begin{remark} \label{supp_remark} The assumptions:  $$\sup_{y \geq 0} \T(y) = 1, \; \inf_{y \geq 0} \T(y) = 0$$ are no stronger than the assumption that $\T$ is a bounded trimming function. In fact, given any arbitary bounded trimming function with $\inf_{y \geq 0} \T(y) = a$ and $\sup_{y \geq 0} \T(y) = b$, the spectral estimator constructed using $\T$ has the same performance as the spectral measure constructed using $$\tilde{\T}(y) \Mydef (\T(y)-a)/(b-a).$$ This is because, 
\begin{align*}
\widetilde{\bm M} \Mydef \bm A^\UH \widetilde{\bm T} \bm A &= \frac{1}{b-a} \bm A^\UH \bm T \bm A - \frac{a}{b-a} \bm I_n \\ &=  \frac{1}{b-a} \bm M - \frac{a}{b-a} \bm I_n. 
\end{align*}
In particular $\bm M$ and $\widetilde{\bm M}$ have the same leading eigenvector. We require the assumption that the trimming function is bounded since a number of results in free probability theory that we rely on assume this. 
\end{remark}

\begin{remark} We need (3) in Assumption \ref{regularity_assumption} to ensure that the limiting spectral measure of the matrix $\bm M$ has no discrete component. We expect that this assumption can be completely removed by a careful analysis since the location of point masses in the limiting spectral measure of $\bm M$ is well understood. 
\end{remark}

\subsection{Main Result}
 In order to state our main result about the performance of the spectral estimator, we need to introduce the following four functions: 
\begin{align}
    \Lambda(\tau)  &\triangleq \tau - \frac{(1-1/\delta)}{\E \left[ \frac{1}{\tau - T} \right]},  \;    \psi_1(\tau)  \triangleq  \frac{\E \left[ \frac{|Z|^2}{\tau - T} \right]}{\E \left[ \frac{1}{\tau - T} \right]}, \nonumber \\
    \psi_2(\tau)  &\triangleq \frac{\E \left[ \frac{1}{(\tau - T)^2} \right]}{ \left(\E \left[ \frac{1}{\tau - T} \right]\right)^2},\; \psi_3^2(\tau)  \Mydef \frac{\E\left[ \frac{|Z|^2}{(\tau - T)^2} \right]}{\left(\E\left[ \frac{1}{\tau - T} \right]\right)^2}. \label{key_functions}
\end{align}
In the above display, the random variables $Z,T$ have the joint distribution given by $Z \sim \cgauss{0}{1}, \; T = \T(|Z|/\sqrt{\delta})$. The functions $\Lambda, \psi_1$ are defined on $[1,\infty)$ and the functions $\psi_2, \psi_3$ are defined on $(1,\infty)$.

\begin{remark}  Under Assumption \ref{regularity_assumption}, the support of the random variable $T$ is the interval $[0,1]$. Hence the definition of these functions at $\tau = 1$ needs some clarification. First, note that the random variable $(1-T)^{-1} \geq 0 $. Hence, the $\E[(1-T)^{-1}]$ is well-defined, but maybe $\infty$.  If it is finite, each of the above functions are well-defined at $\tau=1$. If $\E[(1-T)^{-1}] = \infty$, we define, $\Lambda(1) = 1, \psi_1(1) = 1$. This corresponds to interpreting $1/\infty = 0$ and $\infty/\infty = 1$ in the definition of these functions.
\end{remark}

\begin{theorem} Define $\tau_{r} \triangleq \arg\min_{\tau \in [1,\infty)} \Lambda(\tau)$. Also, let $\theta_\star$ denote the unique value of $\theta> \tau_{r}$ that satisfies $\psi_1(\theta) = \frac{\delta}{\delta - 1}$. Then, under Assumption \ref{regularity_assumption}, we have
	\begin{align*}
	\lambda_1(\bm M) \explain{a.s.}{\rightarrow} \begin{cases}   \Lambda(\tau_{r}), &\psi_1(\tau_{r}) \leq \frac{\delta}{\delta - 1}, \\
	\Lambda(\theta_\star),  & \psi_1(\tau_{r}) > \frac{\delta}{\delta - 1}.
	\end{cases}
	\end{align*}
	Furthermore, 
	\begin{align*}
	\frac{|\bm x_\star^\UH  \hat{\bm x}|^2}{n} \explain{a.s.}{\rightarrow}  \begin{cases} 0, & \psi_1(\tau_{r}) < \frac{\delta}{\delta - 1}, \\
	\frac{\left(\frac{\delta}{\delta -1 } \right)^2 - \frac{\delta}{\delta -1 } \cdot \psi_2(\theta_\star)}{\psi_3(\theta_\star)^2 - \frac{\delta}{\delta -1 } \cdot \psi_2(\theta_\star)}, & \psi_1(\tau_{r}) > \frac{\delta}{\delta - 1} .
	\end{cases}
	\end{align*}
	\label{result_lambda1}
\end{theorem}

\begin{remark} The proof of Theorem \ref{result_lambda1} shows that if $\psi_1(\tau_r) > \delta/(\delta - 1)$, there exists \emph{exactly} one solution to the equation $\psi_1(\theta) = \delta/(\delta -1),\; \theta \in (\tau_r,\infty)$. Hence, $\theta_\star$ is well-defined.
\end{remark}
The proof of this result is postponed until Section \ref{proof_section}. Before we proceed to the proof of this theorem, let us clarify some of its interesting features. First, note that similar to the Gaussian sensing matrices, even in the case of partial orthogonal matrices, the maximum eigenvector exhibits a phase transition behavior. For certain values of $\delta>1$, the inequality $\psi_1(\tau_r) < \frac{\delta}{\delta - 1}$ holds, and hence the maximum eigenvector does not carry information about $\bm{x}_*$. For other values of $\delta$, the inequality $\psi_1(\tau_\star) > \frac{\delta}{\delta - 1}$ holds and hence, the direction of the maximum eigenvector starts to offer information about the direction of $\bm{x}_*$. For typical choices of the trimming function $\T$, there exists a critical value of $\delta$, denoted by $\delta_\T$ such that, when $\delta < \delta_\T$, the spectral estimator is asymptotically orthogonal to the signal vector. When $\delta > \delta_\T$, the spectral estimator makes a non-trivial angle with the signal vector. This phase transition phenomena is illustrated in Figure \ref{phase_transition_plot} for 3 different choices of $\T$.

\begin{figure}[!t]
\centering
\includegraphics[width=3in]{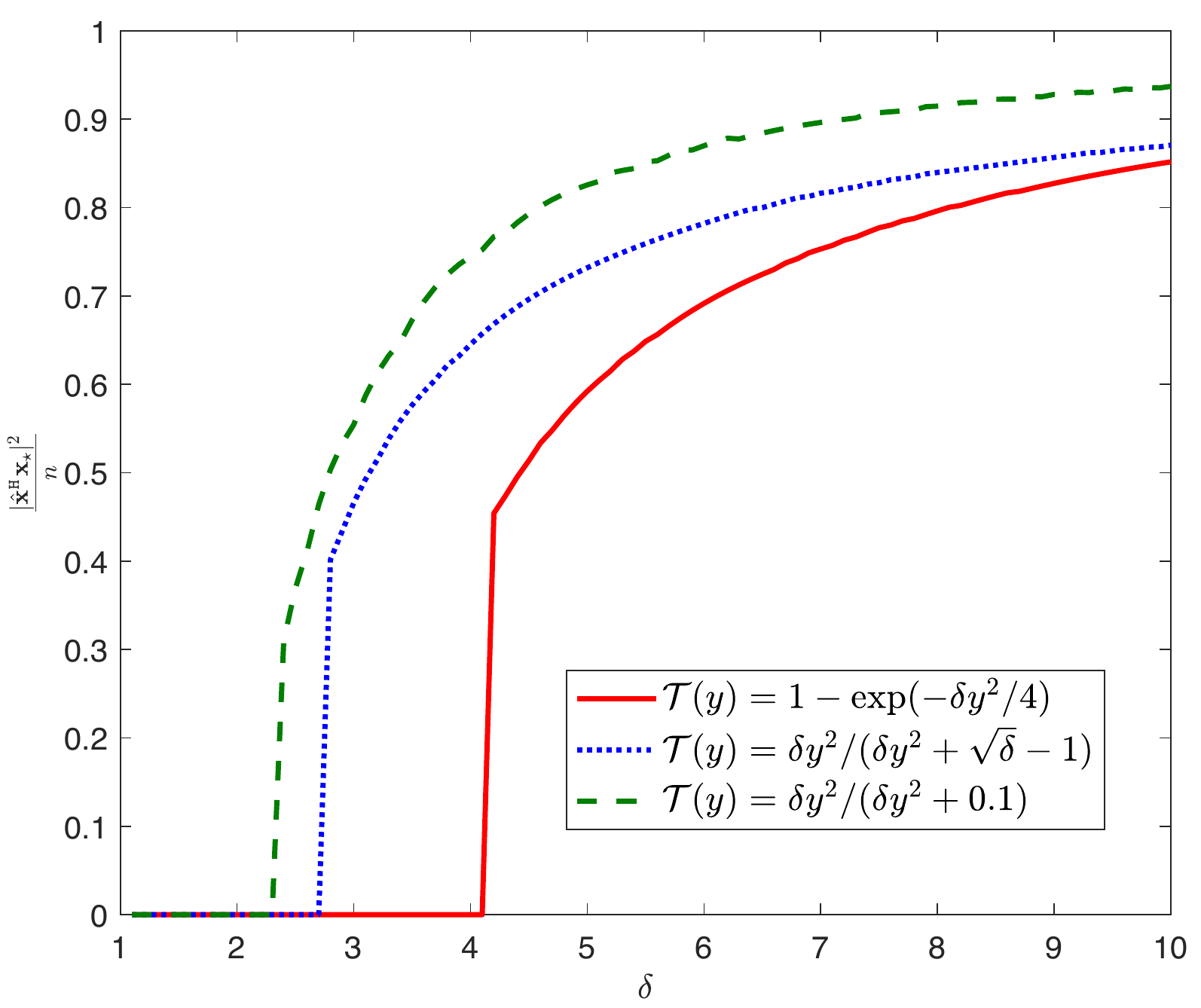}
\caption{Plot of the asymptotic cosine similarity between $\hat{\bm x}$ and $\bm x_\star$.}
\label{phase_transition_plot}
\end{figure}

\begin{remark}[Choice of Trimming function] The trimming function in Figure \ref{phase_transition_plot} are supported on $[0,1]$.
	\begin{enumerate}
		\item $\T(y) = \delta y^2/(\delta y^2 + \sqrt{\delta} - 1)$ is a translated and re-scaled version of the trimming function proposed by \cite{mondelli2017fundamental}. 
		\item $\T(y) = \delta y^2/(\delta y^2 + 0.1)$ is a regularized version of the trimming function proposed by \cite{Lu2018}.
	\end{enumerate}
\end{remark}

\begin{remark}[Extensions to generalized linear measurements] While we focus on the phase retrieval problem in this paper, our results extend straightforwardly to the generalized linear estimation, where the measurements $y_i$ are generated as follows:
\begin{align*}
    y_i & \sim f(\cdot | (\bm A \bm x_\star)_i),
\end{align*}
where $f(\cdot | \cdot)$ denotes a conditional distribution modelling a possibly randomized output channel. Under suitable regularity assumptions on $f$, Theorem 1 holds with the change that the joint distribution of the random variables $T,Z$ is now given by:
\begin{align*}
    Z \sim \cgauss{0}{1}, \; Y \sim f \bigg( \cdot \bigg| \frac{Z}{\sqrt{\delta}} \bigg),  \;   T  = \T(Y).
\end{align*}
\end{remark}

\section{Optimal Trimming Functions}
Theorem \ref{result_lambda1} can used to design the trimming function $\T$ optimally in order to obtain the best possible value of $|\bm x_\star^\UH \hat{\bm x}|^2$. Most of the work towards this goal was already done in \cite{ma2019spectral} where the result in Theorem \ref{result_lambda1} was stated as a conjecture and was used to design the optimal trimming function. In particular, \cite{ma2019spectral} showed the following impossibility result. 

\begin{proposition}[\cite{ma2019spectral}] \label{junjie_impossibility_result} Let $\T$ be any trimming function for which Theorem \ref{result_lambda1} holds. Then, 
\begin{align*}
    \limsup_{\substack{m,n \rightarrow \infty \\m  = n\delta}} \;  \frac{|\bm x_\star^\UH \hat{\bm x}|^2}{n}  &\explain{a.s.}{\leq} \rho^2_\opt(\delta), \; 
\end{align*}
where, 
\begin{align*}
    \rho^2_\opt(\delta) & \Mydef \begin{cases} 0,  & \delta \leq 2 \\ \frac{\theta_\star^\opt - 1}{\theta_\star^\opt - \frac{1}{\delta}}, & \delta > 2  \end{cases},
\end{align*}
where $\theta_\star^\opt$ is the solution to the equation (in $\tau$):
\begin{align*}
    \psi_1^\opt(\tau) & = \frac{\delta}{\delta - 1}, \; \psi_1^\opt(\tau) \Mydef  \frac{\E \left[ \frac{|Z|^2}{\tau - T_\opt} \right]}{\E \left[ \frac{1}{\tau - T_\opt} \right]}, \; \tau \in (1,\infty),
\end{align*}
which exists uniquely when $\delta > 2$ and, the random variable $T_\opt$ is distributed as:
\begin{align*}
    Z & \sim \cgauss{0}{1}, \; T_\opt = 1 - \frac{1}{|Z|^2 }.
\end{align*}
\end{proposition}
The work \cite{ma2019spectral} also provided a candidate for the optimal trimming function:
\begin{align*}
    \T_\opt(y) & = 1 - \frac{1}{\delta y^2}.
\end{align*}
They showed that if the characterization given in Theorem \ref{result_lambda1} holds for $\T_\opt$, then it achieves the asymptotic squared correlation $\rho^2_\opt(\delta)$. Unfortunately, since $\T_\opt$ is unbounded, Theorem~\ref{result_lambda1} does not apply to it. Extending Theorem \ref{result_lambda1} to unbounded trimming functions would likely require extending previously known results in free probability to unbounded measures, and we don't pursue this approach in our work. Instead, we suitably modify the arguments of \cite{ma2019spectral} to show that the family of bounded trimming functions:
\begin{align*}
    \T_{\opt,\epsilon}(y) & = 1 - \frac{1}{\delta y^2 + \epsilon}, \; \epsilon > 0,
\end{align*}
attains an asymptotic squared correlation that can be made arbitrarily close to $\rho^2(\delta)$ as $\epsilon \downarrow 0$.

\begin{proposition} \label{optimal_prop} Let $\hat{\bm x}_\epsilon$ denote the spectral estimator for $\bm x_\star$ obtained by using $\T_{\opt,\epsilon}$ as the trimming function. We have, almost surely,
\begin{align*}
    \lim_{\epsilon \downarrow 0} \; \lim_{\substack{m,n \rightarrow \infty\\ m = n \delta}} \; \frac{|\bm x_\star^\UH \hat{\bm x}_\epsilon|^2}{n} & = \rho^2_\opt(\delta).
\end{align*}
\end{proposition}
We provide a proof of this result in Appendix \ref{proof_optimal_trimming}.

The regularized trimming functions $\T_{\opt,\epsilon}$ are not only useful from a theoretical point of view to prove an achievability result, but also from a computational stand point: In simulations we have observed that the power iterations are slow to converge when $\T_{\opt}$ is used as the trimming function due to presence of large negative eigenvalues and this problem is mitigated by using $\T_{\opt,\epsilon}$ with a small value of $\epsilon$ (such as $0.1$ or $0.01$) with a negligible degradation in performance. 
\section{Proof of Theorem \ref{result_lambda1}}  \label{proof_section}

\subsection{Roadmap}
Our proof follows the general strategy taken by \cite{Lu17}. In this subsection, we state several key lemmas and show how they fit together in the proof of Theorem \ref{main_result}.  First we note that  without loss of generality, for the purpose of analysis of the spectral estimator, we can assume $\bm x_\star = \sqrt{n} \bm e_1$. The following lemma supports this claim. 

\begin{lemma} The distribution of the cosine similarity, $\rho^2 = |\bm x_\star ^\UH \hat{\bm x}|^2/n$ is independent of $\bm x_\star$.
\end{lemma}
\begin{proof}
	
	Let $\bm x_\star$ be an arbitrary signal vector with $\|\bm x_\star \|= \sqrt{n}$. Let $\bm y, \bm T, \hat{\bm x}$ denote the measurements, trimmed measurements and spectral estimate generated when the sensing matrix was $\bm A$ and the signal vector was $\bm x_\star$. Note that the cosine similarity $\rho^2$ is a (deterministic) function of $\bm A, \bm x_\star$ and hence we use the notation $\rho^2(\bm A, \bm x_\star)$ to denote the cosine similarity when the sensing matrix is $\bm A$ and the signal vector is $\bm x_\star$.
	
	Let $\bm \Gamma \in \mathbb{U}(n)$ be  such that $\sqrt{n} \bm \Gamma \bm e_1 = \bm x_\star$. We have $ \bm x_\star ^\UH  \hat{ \bm x} = \sqrt{n} \bm e_1^\UH \bm \Gamma^\UH \hat{\bm x}$. Next we note that $\hat{\bm x}^\prime \Mydef \bm \Gamma^\UH \hat{\bm x} $ is the leading eigenvector of the matrix $ \bm M^\prime \Mydef \bm \Gamma^\UH \bm M \bm \Gamma = (\bm A \bm \Gamma)^\UH \bm T \bm  A \bm \Gamma = {\bm A^\prime}^\UH \bm T \bm A^\prime$, where we defined $\bm A^\prime \Mydef \bm A \bm \Gamma$. Noting that $\bm T$ is a diagonal matrix consisting of the trimmed observations $\bm y = |\bm A \bm x_\star| = \sqrt{n}|\bm A^\prime \bm e_1|$, we conclude that $\hat{\bm x}^\prime$ is the spectral estimate generated when the sensing matrix was $\bm A^\prime$ and the signal vector was $\sqrt{n} \bm e_1$. Hence, we have concluded that
	\begin{align*}
	\rho^2(\bm A, \bm x_\star) & = \rho^2(\bm A^\prime, \sqrt{n} \bm e_1).
	\end{align*}
	Next we note that $\bm A$ was generated from the sub-sampled Haar model, that is $\bm A = \bm H_m \bm S_{m,n}$ where ${\bm H_m \sim \Unif(\mathbb U(m))}$.
	Since the Haar measure on $\mathbb{U}(n)$ is invariant to right multiplication by unitary matrices, we have
	\begin{align*}
	\bm H_m \explain{d}{=} \bm H_m \cdot \begin{bmatrix} \bm \Gamma & \bm 0 \\ \bm 0 & \bm I_{m-n} \end{bmatrix},
	\end{align*}
	where the notation $\explain{d}{=}$ means that two random vectors have the same distributions. Consequently $\bm A = \bm H_m \bm S_{m,n} \explain{d}{=} \bm A \bm \Gamma = \bm A^\prime$.  Therefore, $\rho^2(\bm A, \bm x_\star) = \rho^2(\bm A^\prime, \sqrt{n} \bm e_1) \explain{d}{=} \rho^2(\bm A, \sqrt{n} \bm e_1)$, and the distribution of $\rho^2$ is independent of $\bm x_\star$.
\end{proof}

In the light of the above lemma, in the rest of the paper, we will assume $\bm x_\star = \sqrt{n} \bm e_1$. Next, we partition $\bm A$ by separating the first column $$\bm A = [\bm A_1, \bm A_{-1}],$$ where $\bm A_{-1}$ denotes all the remaining columns of $\bm A$ (except $\bm{A}_1$). Hence we can partition $\bm A^\UH \bm T \bm A$ in the following way:
\begin{align}
\bm A^\UH \bm T \bm A & = \begin{bmatrix} \bm A_1^\UH \bm T \bm A_1 & \bm A_1^\UH \bm T \bm A_{-1} \\
\bm A_{-1}^\UH \bm T \bm A_1 & \bm A_{-1}^\UH \bm T \bm A_{-1}
\end{bmatrix}.
\label{partition_eq}
\end{align}
Our strategy will be to reduce questions about the spectrum of the matrix $\bm M$ to questions about the spectrum of a matrix of the form $\bm X = \bm E \bm U \bm F \bm U ^\UH$, where $\bm U$ is a uniformly random unitary matrix, $\bm E$ is a random matrix independent of $\bm U$ and $\bm F$ is deterministic. This matrix model has been well studied in Free Probability \cite{belinschi2017outliers}.
The starting point of our reduction is Proposition 2 from \cite{Lu17}, stated below. 
\begin{proposition}[\cite{Lu17}] Let $\bm D$ be an arbitrary deterministic symmetric matrix partitioned as:
	\begin{align*}
	\bm D & = \begin{bmatrix} a & \bm q^\UH \\
	\bm q & \bm P
	\end{bmatrix}.
	\end{align*}
	Then, we have
	\begin{align*}
	\lambda_1(D) & = L(\vartheta_\star),
	\end{align*}
	where $L(\vartheta) = \lambda_1(\bm P + \vartheta \bm q \bm q^\UH)$, and $\vartheta_\star>0$ is the unique solution to the fixed point equation 
	$L(\vartheta)  = \frac{1}{\vartheta} + a$.
	 Furthermore, let $\bm v_1$ be the eigenvector corresponding to the largest eigenvalue of $\bm D$. Then,
	\begin{align*}
	|\bm e_1^\UH \bm v_1|^2 & \in \left[\frac{\partial_{-}L(\vartheta_\star)}{\partial_{-}L(\vartheta_\star) + (1 / \vartheta_\star)^2}, \frac{\partial_{+}L(\vartheta_\star)}{\partial_{+}L(\vartheta_\star) + (1 / \vartheta_\star)^2}\right],
	\end{align*}
	where $\partial_{-}$ and $\partial_+$ denote the left and right derivatives respectively. In particular, if $L(\vartheta)$ is differentiable at $\vartheta_\star$, then
	\begin{align*}
	|\bm e_1^\UH \bm v_1|^2 = \frac{L'(\vartheta_\star)}{L'(\vartheta_\star) + (1 / \vartheta_\star)^2}.
	\end{align*}
	\label{lu_reduction}
\end{proposition}

A straightforward corollary of the above proposition to our problem is given below. Define the function
\begin{align*}
L_m(\vartheta) & \Mydef \lambda_1 \left( \bm A_{-1}^\UH( \bm T +  \vartheta \bm T \bm A_1 (\bm T \bm A_1)^\UH) \bm A_{-1}  \right).
\end{align*}

\begin{corollary} 
	Let $\vartheta_m>0$ be the unique solution of $L_m(\vartheta) = 1/\vartheta + \bm A_1^\UH \bm T \bm A_1$. Then, $\lambda_1(\bm A^\UH \bm T \bm A) = L_m(\vartheta_m)$ and
	\begin{align*}
	|\bm e_1^\UH \hat{\bm x}|^2 & \in \left[\frac{\partial_{-}L_m(\vartheta_m)}{\partial_{-}L_m(\vartheta_m) + (1 / \vartheta_m)^2}, \frac{\partial_{+}L_m(\vartheta_m)}{\partial_{+}L(\vartheta_m) + (1 / \vartheta_m)^2}\right].
	\end{align*}
	In particular, if  $L_m(\vartheta)$ is differentiable at $\vartheta_m$, then
	\begin{align*}
	|\bm e_1^\UH \hat{\bm x}|^2 = \frac{L_m'(\vartheta_m)}{L_m'(\vartheta_m) + (1 / \vartheta_m)^2}.
	\end{align*}
	\label{lu_corollary}
\end{corollary}
Hence, we shift our focus to characterizing the function $L_m(\vartheta)$. Recall the decomposition of the matrix $\bm M$ given in (\ref{partition_eq}). Recall that since $\bm x_\star = \sqrt{n} \bm e_1$, the diagonal matrix $\bm T$ is a deterministic function of $\bm A_1$. If the sensing matrix $\bm A$ consisted of independent Gaussian entries, then $\bm T, \bm A_1$ would have been independent of $\bm A_{-1}$. This is no longer true when $\bm A$ is a partial unitary matrix. In order to take care of this, the following lemma leverages a conditioning trick to get rid of the dependence. The following lemma also establishes the link between the function $L_m(\vartheta)$ and the study of the spectrum of a matrix of the form $\bm X = \bm E \bm U \bm F \bm U ^\UH$, where $\bm U$ is a uniformly random unitary matrix, $\bm E$ is a random matrix independent of $\bm U$ and $\bm F$ is deterministic.
\begin{lemma}
	We have
	\begin{align}
	L_m(\vartheta) & = \lambda_1 \left( \bm B^\UH (\bm T + \vartheta \bm T \bm A_1 (\bm T \bm A_1)^\UH) \bm B \bm H_{m-1} \bm R \bm H_{m-1}^\UH\right),
	\label{L_m_eq_numbered}
	\end{align}
	where
	\begin{align*}
	\bm R & = \begin{bmatrix} \bm I_{n-1} & \bm 0_{n-1,m-n} \\ \bm 0_{m-n,n-1} & \bm 0_{m-n,m-n}
	\end{bmatrix},
	\end{align*}
	$\bm B \in \mathbb{C}^{m \times m-1}$ is an arbitrary  basis matrix for $\bm A_1^\perp$, which denotes the subspace orthogonal to $\bm A_1$,  and $\bm H_{m-1} \sim \Unif(\mathbb{U}(m-1))$ is independent of $\bm A_1$.
\end{lemma}
\begin{proof} We condition on $\bm A_1$. Conditioned on $\bm A_1$, we can realize $\bm A_{-1}$ as: 
	\begin{align*}
	\bm A_{-1} & = \bm B\bm H_{m-1} \bm S_{m-1,n-1}.
	\end{align*}
	In the above equation, $\bm B \in \mathbb C^{m \times m-1}$ is matrix whose columns form an orthonormal basis of the orthogonal complement of $\bm A_1$ and $\bm H_{m-1}$ is a Haar Unitary of size $m-1$ independent of $\bm A_1$. Hence, we obtain
	\begin{align*}
	&L_m(\vartheta)  = \lambda_1 \left( \bm A_{-1}^\UH (\bm T + \vartheta \bm T \bm A_1 (\bm T \bm A_1)^\UH) \bm A_{-1} \right) \\
	&\explain{a}{=} \lambda_1 \left( \bm B^\UH (\bm T + \vartheta \bm T \bm A_1 (\bm T \bm A_1)^\UH) \bm B \cdot \bm H_{m-1} \bm R \bm H_{m-1}^\UH\right). 
	\end{align*}
	In the step marked (a), We used the fact that for any two matrices $\bm \Lambda, \bm \Gamma$ (of appropriate dimensions), $\bm \Lambda \bm \Gamma$ and $\bm \Gamma \bm \Lambda$ have the same non-zero eigenvalues. In particular, we used this fact with:
	\begin{align*}
	    \bm \Lambda & = \bm S_{m-1,n-1}^\UH \bm H_{m-1}^\UH \\
	    \bm \Gamma & = \bm B^\UH (\bm T + \vartheta \bm T \bm A_1 (\bm T \bm A_1)^\UH) \bm B \bm H_{m-1} \bm S_{m-1,n-1}.
	\end{align*}
\end{proof}
Define the matrix,
\begin{align}
\bm E(\vartheta) \Mydef \bm B^\UH (\bm T + \vartheta \bm T \bm A_1 (\bm T \bm A_1)^\UH) \bm B.
\end{align}
The following lemma characterizes the asymptotic limit of the function $L_m(\vartheta)$.
Define $\Lambda_+(\tau)$ as
\begin{align*}
\Lambda_+(\tau) & = \begin{cases} \tau - \frac{(1-1/\delta)}{\E\left[\frac{1}{\tau- T}\right]} \ \  \ \ \ \ \ \ \ \ \ \ \ \ \ \   {\rm if} \ \tau > \tau_{r},  \\
\min_{\tau \geq 1} \left(  \tau - \frac{(1-1/\delta)}{\E\left[\frac{1}{\tau- T}\right]} \right) \ \ {\rm if} \ \tau \leq \tau_{r},
\end{cases}
\end{align*}
where $T=\T(|Z|/\sqrt{\delta})$ and $Z\sim\mathcal{CN}(0,1)$, and
\begin{align*}
\tau_{r} & \triangleq \arg \min_{\tau \geq 1} \left(  \tau - \frac{(1-1/\delta)}{\E\left[\frac{1}{\tau- T}\right]} \right).
\end{align*}

\begin{lemma}\label{L_m_func_asymptotic}
Let $\vartheta_c  \Mydef \left(1 - \left( \E \left[ \frac{|Z|^2}{1-T} \right] \right)^{-1} - \E[|Z|^2 T]  \right)^{-1}$. Define the function $\theta(\vartheta)$ as:
	\begin{itemize}
		\item {When $\vartheta> \vartheta_c$: } Let $\theta(\vartheta)$ be the unique value of $\lambda$ that satisfies the equation: $$\lambda - \E[|Z|^2 T] - 1/\vartheta  = \left(\E \left[ \frac{|Z|^2}{\lambda - T} \right]\right)^{-1},$$ in the interval:
		\begin{align*}
		    \lambda \in  \left( \max(1, \E[ |Z|^2 T]+1/\vartheta), \infty \right).
		\end{align*}
		\item {When $\vartheta \leq \vartheta_c$: } $\theta(\vartheta) \Mydef 1$.
	\end{itemize}
Then, we have $L_m(\vartheta) \explain{a.s.}{\rightarrow} \Lambda_{+} (\theta(\vartheta))$, where $L_m(\vartheta) $ is defined in \eqref{L_m_eq_numbered}.
\end{lemma}

 The proof of Lemma \ref{L_m_func_asymptotic} can be found in Section \ref{Sec:technical_lemmas}.

From Corollary \ref{lu_corollary}, we know that $\lambda_1(\bm M)$ solves the fixed point equation (in $\vartheta$): $L_m(\vartheta) = 1/\vartheta + \bm A_1^\UH \bm T \bm A_1.$ Simple concentration arguments (see Lemma \ref{concentration_result}, Section \ref{Sec:spectrum_E}) show that asymptotically: 
\begin{align*}
    \bm A_1^\UH \bm T \bm A_1 \approx \E |Z|^2 T.
\end{align*}
Combining this with Lemma \ref{L_m_func_asymptotic} suggests that asymptotically $\lambda_1(\bm M)$ behaves like the solution to the following fixed point equation (in $\vartheta$):
\begin{align*}
    \Lambda_{+} (\theta(\vartheta)) = 1/\vartheta + \E |Z|^2 T.
\end{align*} 
The following lemma analyzes the behavior of this asymptotic fixed point equation. The proof of this lemma can be found in Section \ref{Sec:technical_lemmas}.
\begin{lemma}\label{Lem:Lambda_plus}
	The following hold for the equation:  $$\Lambda_{+}(\theta(\vartheta)) = 1/\vartheta + \E[|Z|^2 T], \; \vartheta>0.$$
	\begin{enumerate}
		\item This equation has a unique solution. 
		\item Let $\vartheta_\star$ denote the solution of the above equation. Then: 
	\end{enumerate}
			\paragraph{Case 1 } If $ \psi_1(\tau_{r}) \leq \frac{\delta}{\delta - 1},$ we have $$\Lambda_+(\theta(\vartheta_\star)) = \Lambda(\tau_r).$$ Furthermore if $\psi_1(\tau_{r}) < \delta/(\delta - 1)$, then,
			\begin{align*}
			\frac{\diff \Lambda_+(\theta(\vartheta))}{\diff \vartheta} \bigg\rvert_{\vartheta = \vartheta_\star} = 0,
			\end{align*}
			\paragraph{Case 2 } If $ \psi_1(\tau_{r}) > \frac{\delta}{\delta - 1},$ we have $$\Lambda_+(\theta(\vartheta_\star)) = \Lambda(\theta_\star),$$ and, 
			\begin{multline*}
			\frac{\diff \Lambda_+(\theta(\vartheta))}{\diff \vartheta} \bigg\rvert_{\vartheta = \vartheta_\star} = \\  \frac{1}{\vartheta_\star^2} \cdot \frac{\delta}{\delta-1} \cdot \left( \frac{\delta}{\delta -1 } - \psi_2(\theta_\star) \right) \cdot \frac{1}{\psi_3^2(\theta_\star) -\frac{\delta^2}{(\delta-1)^2}}.
			\end{multline*} where $\theta_\star > 1$ is the unique  $\theta \geq \tau_{r}$ that satisfies
			$\psi_1(\theta) = \frac{\delta}{\delta - 1}.$
\end{lemma}

We are now in the position to prove our main result (restated below for convenience). Recall the definitions of the functions $\Lambda(\tau), \psi_1(\tau), \psi_2(\tau), \psi_3(\tau)$ from  Section  \ref{main_result}.\vspace{5pt}

\noindent\textbf{Theorem 1} \textit{Define $\tau_{r} \triangleq \arg\min_{\tau \in [1,\infty)} \Lambda(\tau)$. Also, let $\theta_\star$ denote the unique value of $\theta> \tau_{r}$ that satisfies $\psi_1(\theta) = \frac{\delta}{\delta - 1}$. Then, we have
	\begin{align*}
	\lambda_1(\bm M) \explain{a.s.}{\rightarrow} \begin{cases}   \Lambda(\tau_{r}), &\text{if }\psi_1(\tau_{r}) \leq \frac{\delta}{\delta - 1}, \\
	\Lambda(\theta_\star),  &\text{if } \psi_1(\tau_{r}) > \frac{\delta}{\delta - 1}.
	\end{cases}
	\end{align*}
Furthermore,
	\begin{align*}
	 |\bm e_1^\UH  \hat{\bm x}|^2 \explain{a.s.}{\rightarrow} \begin{cases} 0, &\text{if } \psi_1(\tau_{r}) < \frac{\delta}{\delta - 1}, \\
	\frac{\left(\frac{\delta}{\delta -1 } \right)^2 - \frac{\delta}{\delta -1 } \cdot \psi_2(\theta_\star)}{\psi_3(\theta_\star)^2 - \frac{\delta}{\delta -1 } \cdot \psi_2(\theta_\star)}, & \text{if }\psi_1(\tau_{r}) > \frac{\delta}{\delta - 1} .
	\end{cases}
	\end{align*}
}
\begin{proof}

		We start with the analysis of the largest eigenvalue. We recall the claim of Corollary \ref{lu_corollary}, which tells us that $\lambda_1(\bm M)$ is given by $L_m(\vartheta_m)$ where $\vartheta_m$ denotes the solution of $L_m(\vartheta) = 1/\vartheta + a_m$ and $a_m = \bm A_1^\UH \bm T \bm A_1$. 
		
		We also know that there exists a probability 1 event $\mathcal{E}$, on which,  $L_m(\vartheta)  \explain{a.s.}{\rightarrow} \Lambda_+(\theta(\vartheta))$ (Lemma \ref{L_m_func_asymptotic}) and $a_m \explain{a.s.}{\rightarrow} \E[|Z|^2 T]$ (see Lemma \ref{concentration_result} in Section \ref{Sec:spectrum_E}). 
		
		We claim that on $\mathcal{E}$, $\vartheta_m \rightarrow \vartheta_\star$, where $\vartheta_\star$ is the solution of the limiting fixed point equation $\Lambda_+(\theta(\vartheta)) = 1/\vartheta + \E[|Z|^2 T]$ (which was analyzed in Lemma \ref{Lem:Lambda_plus}). To see this let $\overline{\vartheta} = \lim\sup \vartheta_{m}$. Consider a subsequence $\vartheta_{m_k} \rightarrow \overline{\vartheta}$. Then applying Lemma 3 (in Appendix E) of \cite{Lu17}, we obtain, 
		\begin{align*}
		0 &= \lim_{k \rightarrow \infty} \left( L_{m_k}(\vartheta_{m_k}) - \frac{1}{\vartheta_{m_k}} - a_{m_k} \right) \\&= \Lambda_+(\theta(\overline{\vartheta})) - \frac{1}{\overline{\vartheta}} - \E|Z|^2 T.
		\end{align*}
		That is, $\overline{\vartheta}$ is also a solution to the limiting fixed point equation $\Lambda_+(\theta(\vartheta)) = 1/\vartheta + \E[|Z|^2 T]$. But since this equation has a unique solution (Lemma \ref{Lem:Lambda_plus}), we have $\lim\sup \vartheta_{m} =  \overline{\vartheta}   = \vartheta_\star$. Likewise, an analogous argument shows $\lim\inf \vartheta_{m}  = \vartheta_\star$.
		
		Now for any realization in the event $\mathcal{E}$, we have,
		\begin{align*}
		\lambda_1(\bm M)  = L_m(\vartheta_m) \explain{(a)}{\rightarrow} \Lambda_+(\theta(\vartheta_\star)).
		\end{align*}
		In the above display, in the step marked (a), we again appealed to Lemma 3 (Appendix E) of \cite{Lu17} and the fact that $\vartheta_m \rightarrow \vartheta_\star$. Finally, appealing to the alternative characterization of $\Lambda_+(\theta(\vartheta_\star))$ given in Lemma \ref{Lem:Lambda_plus} gives us the claim of the theorem.

We now discuss our result about the cosine similarity. We recall that from Corollary \ref{lu_corollary}, we have
		\begin{align*}
		|\bm e_1^\UH \hat{\bm x}|^2 & \in \left[\frac{\partial_{-}L_m(\vartheta_m)}{\partial_{-}L_m(\vartheta_m) + (1 / \vartheta_m)^2}, \frac{\partial_{+}L_m(\vartheta_n)}{\partial_{+}L(\vartheta_m) + (1 / \vartheta_m)^2}\right].
		\end{align*}
		Appealing to Lemma 4 in Appendix E of \cite{Lu17}, we have,  
		\begin{align*}
		\partial_{-}L_m(\vartheta_m) \rightarrow \partial_{-} \Lambda_{+}(\theta(\vartheta_\star)), \;  
		\partial_{+}L_m(\vartheta_m) \rightarrow \partial_{+} \Lambda_{+}(\theta(\vartheta_\star)).
		\end{align*}
		The derivative of $\Lambda_+(\theta(\vartheta))$ at $\vartheta = \vartheta_\star$ was calculated in Lemma \ref{Lem:Lambda_plus}. Plugging this in the above expression gives the statement of the theorem.
\end{proof}

The remainder of this section is dedicated to the proof of Lemmas \ref{L_m_func_asymptotic} and \ref{Lem:Lambda_plus}, and is organized as follows:
\begin{itemize}
\item Recall that (cf. \ref{L_m_eq_numbered})
\[
L_m(\vartheta)= \lambda_1\left(\bm{E}(\vartheta)\bm{H}_{m-1}\bm{R}\bm{H}_{m-1}^\UH\right),
\]
where
\[
\bm{E}(\vartheta)\Mydef \bm B^\UH (\bm T + \vartheta \bm T \bm A_1 (\bm T \bm A_1)^\UH) \bm B.
\]
Note that $\bm{E}(\vartheta)$ is independent of $\bm{H}_{m-1}$. The spectrum of such a matrix product has been studied in free probability theory, and we collect some results regarding this in Section \ref{Sec:free}.
\item In order to apply the free probability results, we need to understand the spectrum of $\bm{E}(\vartheta)$. This is done in Section \ref{Sec:spectrum_E}.
\item It turns out that the limiting spectrum measure of $\bm{E}(\vartheta)\bm{H}_{m-1}\bm{R}\bm{H}_{m-1}^\UH$ is given by the free convolution (defined in Section \ref{Sec:free}) of the measures $\gamma$ and $\mathcal{L}_T$, where $\gamma\Mydef \frac{1}{\delta}\delta_1 + \left(1-\frac{1}{\delta}\right)\delta_0$ and $\mathcal{L}_T$ is the law of the random variable $T=\T(|Z|/\sqrt{\delta})$. Section \ref{Sec:support} is devoted to understanding the support of the free convolution.
\item Finally, Section \ref{Sec:technical_lemmas} proves lemmas \ref{L_m_func_asymptotic} and \ref{Lem:Lambda_plus}.
\end{itemize}
\subsection{Free Probability Background}\label{Sec:free}

	Our analysis of the spectral estimators relies on a well-studied model in the theory of free probability; We will reduce the problem to the problem of understanding the spectrum of  matrices of the form $\bm X = \bm E \bm U \bm F \bm U^\UH$, where $\bm E$ and $\bm F$ are deterministic matrices and $\bm U$ is a Haar-distributed unitary matrix. Then, the limiting spectral distribution of $\bm X$ is the \emph{free multiplicative convolution} of the limiting spectral distributions of $\bm E$ and $\bm F$. This section is a collection of the results and definitions regarding these aspects. Here is the organization of this section. Section \ref{free_harmonic_analysis} collects various facts from free harmonic analysis. 
	Section \ref{RMT_model_X} describes the two fundamental results about the model $\bm X = \bm E \bm U \bm F \bm U^\UH$ that will be used throughout our paper. Section \ref{singular_part_free_convolution} reviews some results about the support of singular part of the free convolution of two measures. Throughout this section, we assume that $\gamma$ and $\nu$ are two arbitrary compactly supported probability measures on $[0,\infty)$ and that neither of the two measures is completely concentrated at a single point.
	
	\subsubsection{Facts from Free Harmonic Analysis}
	\label{free_harmonic_analysis}
	In this section, we collect some facts from the field of free harmonic analysis. All these results can be found in Chapter 3 of \cite{mingo2017free} or the papers \cite{belinschi2017outliers} and \cite{belinschi2003atoms}.
	\begin{definition} The Cauchy transform $G_\gamma$ of $\gamma$ at $z$ is defined as follows: 
		\begin{align*}
		G_\gamma(z) & = \int \frac{\gamma(\diff t)}{z - t}, \; z \in \mathbb{C} \backslash [0,\infty).
		\end{align*}
	\end{definition}	
	\begin{definition}	
		The moment generating function of $\gamma$, $\psi_\gamma$ at $z$ is defined as follows:
		\begin{align*}
		\psi_\gamma(z) & = \int \frac{ zt }{1 - zt} \gamma(\diff t), \; z \in \mathbb{C} \backslash [0,\infty).
		\end{align*}
	\end{definition}
	The Cauchy transform and the moment generating function are related via the relation
	\begin{align*}
	G_\gamma(z) & = \frac{1}{z} \cdot \left( \psi_\gamma \left( \frac{1}{z} \right) + 1 \right).
	\end{align*}
	\begin{definition}
		The $\eta$-transform of a measure is defined as,
		\begin{align*}
		\eta_\gamma(z) & = \frac{\psi_\gamma(z)}{1+\psi_\gamma(z)}.
		\end{align*}
	\end{definition}
	The Cauchy Transform (and hence the Moment Generating function) uniquely characterizes a measure. The measure can be obtained by the following inversion formula. The particular version we state is taken from Section 3.1 of  \cite{belinschi2017outliers}.
	\begin{theorem} For $a<b \in [0,\infty)$, we have
		\begin{align*}
		\gamma((a,b)) + \frac{1}{2} \gamma(\{a,b\}) & = \frac{1}{\pi} \lim_{\epsilon \rightarrow 0^+} \int_a^b \Im(G_\gamma(x-i \epsilon)) \diff x.
		\end{align*}
		Furthermore, if $\gamma$ satisfies $\gamma = \gamma_{ac} + \gamma_{s}$, where $\gamma_{ac}$ and $\gamma_s$ denote the absolutely continuous and the singular part of the measure with respect to the Lebesgue measure, then the density of the absolutely continuous part is given by
		\begin{align*}
		\frac{\diff \gamma_{ac}}{\diff x}(x) & = \lim_{\epsilon \rightarrow 0^+} \frac{1}{\pi}\Im(G_\gamma(x-i\epsilon)).
		\end{align*}
		\label{stieltjes_theorem}
	\end{theorem}
	Next we recall the definition of  the free convolution based on the subordination functions from \cite{belinschi2007new}. The statement we provide below appears in a more general form as Proposition 2.6 in \cite{belinschi2014operator}.
	
	\begin{definition} Let $(\gamma,\nu)$ be a pair of probability measures. There exist analytic functions $w_\gamma, w_\nu$ defined  on $\mathbb C \backslash [0,\infty)$ such that, for all $z \in \mathbb C^{+}$ we have
		\begin{enumerate}
			\item $w_\gamma(z), w_\nu(z) \in \mathbb C^{+}$; $w_\gamma(\overline{z}) = \overline{ w_\gamma(z)}, w_\nu(\overline{z}) = \overline{ w_\nu(z)}$ and $\Arg(w_\gamma(z)) \geq \Arg(z), \Arg(w_\nu(z)) \geq \Arg(z)$.
			
			\item For any $z \in \mathbb C^+$, $w_\nu(z)$ is the unique solution in $\mathbb C^{+}$ of the fixed point equation $Q_z(w) = w$, where $Q_z$ is given by
			\begin{align*}
			Q_z(w) & = \frac{w}{\eta_\nu(w)} \eta_\gamma \left( \frac{z \eta_\nu(w)}{w} \right).
			\end{align*}
			An analogous characterization holds for $w_\gamma$ with the role of $\gamma$ and $\nu$ changed.
		\end{enumerate}
		The free convolution of the measures $\gamma$ and $\nu$ denoted by $\gamma \boxtimes \nu$ is the measure whose moment generating function satisfies
		$$\psi_{\gamma \boxtimes \nu}(z) = \psi_{\gamma}(w_\gamma(z)) = \psi_{\nu}(w_\nu(z)) = \frac{w_\gamma(z) w_\nu(z)}{z - w_\gamma(z) w_\nu(z)}.  $$
		\label{def_subordination}
	\end{definition}
	
	\begin{remark} We emphasize that each of the subordination functions $w_{\gamma}, w_{\nu}$ depend on both the measures $\gamma, \nu$. This is clear since the function $Q_z(w)$ defining $w_\nu$ depends on both $\nu,\gamma$.	
	\end{remark}
	
	Note that the above definition defines $w_\nu$ and $w_\gamma$ on $\mathbb C \backslash [0,\infty)$. However these functions can be continously extended to $\overline{\mathbb C^{+}} \cup \{\infty\}$ (Lemma 3.2 in \cite{belinschi2017outliers}). These extensions to the real line will be important for Theorem \ref{RMT_model_X}.
	
	\begin{lemma} The restrictions of subordination functions $w_\gamma, w_\nu$ on $\mathbb C^{+}$ have extensions to $\overline{\mathbb C^{+}} \cup \{\infty\}$ with the following properties:
		\begin{enumerate}
			\item $w_\gamma, w_\nu: \overline{\mathbb C^{+}} \cup \{\infty\} \rightarrow \overline{\mathbb C^{+}} \cup \{\infty\}$ are continuous.
			\item If $1/x \in [0,\infty) \backslash \text{Supp}(\gamma \boxtimes \nu)$, then the functions $w_\gamma, w_\nu$ continue analytically to a neighborhood of $x$ and
			\begin{align*}
			\frac{1}{w_\gamma(x)} & = \frac{w_\nu(x)}{x} \cdot \frac{1+\psi_\nu(w_\nu(x))}{\psi_\nu(w_\nu(x))} \in \mathbb R \backslash \text{Supp}(\gamma), \\
			\frac{1}{w_\nu(x)} & = \frac{w_\gamma(x)}{x} \cdot \frac{1+\psi_\gamma(w_\gamma(x))}{\psi_\gamma(w_\gamma(x))} \in \mathbb R \backslash \text{Supp}(\nu).
			\end{align*}
		\end{enumerate}
		\label{extension_lemma}
	\end{lemma}
	
	\subsubsection{Spectrum of $\mathbf{X} = \mathbf{E} \mathbf{U} \mathbf{F} \mathbf{U}^\UH $} 
	\label{RMT_model_X}
	As we discussed before, we will convert the problem of analyzing the spectrum of $\bm{M}$ to problems involving the spectrum of matrices of the form $\mathbf{X}_N = \mathbf{E}_N \mathbf{U}_N \mathbf{F}_N \mathbf{U}_N^\UH$, where $\bm U_N$ is a sequence of Haar distributed $N \times N$ random matrices, and $\bm E_N$ and $\bm F_N$ are sequences of deterministic positive semidefinite matrices. In this section, we review two important results from the field of free probability regarding such matrices. 
	
	Suppose that  $\bm E_N$ and $\bm F_N$ satisfy the following hypotheses:
	\begin{enumerate}
		\item[(i)] $\mu_{\bm E_N} \explain{d}{\rightarrow} \mu_e$ and $\mu_{\bm F_N} \explain{d}{\rightarrow} \mu_f $, where $\mu_e,\mu_f$ are compactly supported measures on $[0,\infty)$.
		\item[(ii)] $\bm E_N$ has a single outlying eigenvalue $\theta$ not contained in $\text{Supp}(\mu_e)$. $\bm F_N$ has no eigenvalues outside $\text{Supp}(\mu_f)$.
		\item [(iii)]The set of eigenvalues of $\bm E_N$ not equal to $\theta$ converge uniformly to $\text{Supp}(\mu_e)$ in the sense,
		\begin{align*}
		\lim_{N \rightarrow \infty}\max_{i: \lambda_i(\bm E_N) \neq \theta}\text{dist}(\lambda_i(\bm E_N), \text{Supp}(\mu_e)) = 0.
		\end{align*}
	\end{enumerate}
	
{Our next theorem characterizes the bulk distribution of $\bm X_N$. The first part of this theorem is
		due to \cite{voiculescu1991limit} and the second and third parts are due to \cite{belinschi2017outliers} (Theorem 2.3).}

	\begin{theorem} Let $w_e$ and $w_f$ denote the subordination functions for the free multiplicative convolution of $\mu_e$ and $\mu_f$. Define
		\begin{align*}
		\tau_e(1/z)  = \frac{1}{w_e(1/z)}, \; K= \text{Supp}(\mu_e \boxtimes \nu_f) \cup \tau_e^{-1}(\theta).
		\end{align*}
		Then we have, almost surely for large enough $N$,
		\begin{enumerate}
			\item $\mu_{\bm X_N} \explain{d}{\rightarrow} \mu_e \boxtimes \mu_f$.
			\item Given $\epsilon>0$, we have $\sigma(\bm X_N) \subset K_\epsilon$, where $K_\epsilon$ is the $\epsilon$-neighborhood of $K$ and $\sigma(\bm X_N)$ denotes the set of eigenvalues of $\bm X_N$.
			\item For any $\rho \in \tau_e^{-1}(\theta)$ such that $\exists \epsilon > 0$ with $(\rho - 2 \epsilon, \rho + 2 \epsilon) \cap K = \{\rho \}$, we have $|\sigma(\bm X_N) \cap (\rho-\epsilon,\rho + \epsilon)| = 1$.
		\end{enumerate}
		\label{outlier_theorem_free_probability}
	\end{theorem}
	
	\begin{remark} The hypothesis in the above theorem can be relaxed (as mentioned in Remark 5.11 of \cite{belinschi2017outliers}) in the following two ways: 1) $\bm E_N$ is random, independent of $\bm U_N$ and $\bm F_N$ is deterministic, provided $\mu_{\bm E_N} \explain{d}{\rightarrow} \mu_e$ occurs almost surely, 2) The spike locations depend on $N$, $\theta_N$ provided $\theta_N \rightarrow \theta$ almost surely.
		\label{outlier_theorem_remark}
	\end{remark}

	\begin{remark} The above theorem is a simplified version of Theorem 2.3 in \cite{belinschi2017outliers} which allows for multiple spikes in both $\bm E_N$ and $\bm F_N$.
	\end{remark}
	\begin{remark} The function $\tau$ might not be invertible. In such cases, $\tau^{-1}(\theta)$ can be a non-singleton set, and hence a single spike in $\bm E_N$ can create multiple spikes in $\bm X_N$. But we will see that this doesn't happen in our problem.
	\end{remark}

	\subsubsection{Singular Part of Free Convolution}
	\label{singular_part_free_convolution}
	In the last section we discussed the bulk distribution of $\bm{X}_N = \bm{E}_N \bm{U}_N \bm{F}_N \bm{U}_N$. The main objective of this section is to mention a result regarding the largest eigenvalue of $\bm{X}_N$. We state regularity results for the singular part of $\gamma \boxtimes \nu$ from \cite{belinschi2006note} (Corollary 3.4) and \cite{belinschi2003atoms} (Theorem 4.1).

	\begin{theorem}[Singular Part of $\gamma \boxtimes \nu$]  Decompose the singular part of $\gamma \boxtimes \nu$ as $(\gamma \boxtimes \nu)_s = (\gamma \boxtimes \nu)_d + (\gamma \boxtimes \nu)_{sc} $ where $(\gamma \boxtimes \nu)_d$ denotes the discrete part and $(\gamma \boxtimes \nu)_{sc}$ denotes the singular continous part. Then we have,
		\begin{enumerate}
			\item There can be at most two atoms. The possible locations of the atoms are: 
			\begin{enumerate}
				\item $0$, with $\gamma \boxtimes \nu(\{0\}) = \max(\gamma(\{0\}),\nu(\{0\}))$.
				\item Any $a \in (0,\infty)$ such that there exist $u,v \in (0,\infty)$ with $uv = a$ and $\gamma(\{u\}) + \nu(\{v\}) > 1$ and we have, $\gamma \boxtimes \nu(\{a\}) = \gamma(\{u\}) + \nu(\{v\}) - 1$. Note that there can be atmost one such $a$.
			\end{enumerate}
			\item Suppose neither of $\gamma,\nu$ is completely concentrated at a single point. We have, $\text{Supp}((\gamma \boxtimes \nu)_{sc}) \subset \text{Supp}((\gamma \boxtimes \nu)_{ac}) $. Hence,
			\begin{align*}
			\text{Supp}(\gamma \boxtimes \nu) & = \text{Supp}((\gamma \boxtimes \nu)_{ac}) \cup \text{Supp}((\gamma \boxtimes \nu)_d).
			\end{align*}
		\end{enumerate}
		\label{regularity_theorem}
	\end{theorem}
\subsection{Analysis of the Spectrum of $\mathbf{E}(\vartheta)$}\label{Sec:spectrum_E}

In order to apply Theorem \ref{outlier_theorem_free_probability}, we need to understand the spectrum of $\bm B^\UH (\bm T + \vartheta \bm T \bm A_1 (\bm T \bm A_1)^\UH) \bm B$. This is done in the following lemma.

\begin{lemma} Let $$T_{(1)} \geq  T_{(2)} \dots \geq T_{(m)}$$ denote the sorted trimmed measurements. Let ${\bm E(\vartheta) \Mydef \bm B^\UH (\bm T + \vartheta \bm T \bm A_1 (\bm T \bm A_1)^\UH) \bm B}$. Then, 
	\begin{enumerate}
		\item The eigenvalues of $\bm E(\vartheta)$ interlace with $T_{(1)}, T_{(2)} \dots T_{(m)}$ in the sense, 
		\begin{align*}
		\lambda_i(\bm E(\vartheta))  &\leq T_{(i-1)}\; \forall \; i=2,3, \dots m, \;  \& \\  \lambda_i(\bm E(\vartheta)) &\geq T_{(i+1)} \; \forall  \; i=1,3, \dots m-1.
		\end{align*}
		\item $\bm E(\vartheta)$ can have at most one eigenvalue bigger than $T_{(1)}$, which (if it exists) is given by the root of the following equation: 
		\begin{align*}
		Q_m(\lambda) = \frac{1}{\lambda - a_m - 1/\vartheta}, \; \lambda > \max(a_m+1/\vartheta, T_{(1)}),
		\end{align*}
		where $Q_m(\lambda)$ is defined as
		\begin{align*}
		Q_m(\lambda)  \Mydef \sum_{i=1}^m \frac{|A_{1i}|^2}{\lambda - T_i}.
		\end{align*}
		\item Furthermore, $\lambda_1(\bm E(\vartheta)) \leq 1 + \vartheta$ and $\lambda_{m-1}(\bm E(\vartheta)) \geq 0$.
	\end{enumerate}
	\label{spectrum_E_empirical}
\end{lemma}
\begin{proof}
	Define the matrix $\bm E(\vartheta) =\bm B^\UH (\bm T + \vartheta \bm T \bm A_1 (\bm T \bm A_1)^\UH) \bm B$. The main trick will be to choose the orthonormal basis matrix $\bm B$ conveniently, which will make our calculations easier. Recall that the columns of  matrix $\bm B$, i.e. $\bm B_1, \bm B_2 \dots \bm B_{m-1}$, span the subspace $\bm A_1^\perp$. Any basis for subspace $\bm A_1^\perp$ can serve as matrix $\bm{B}$. Hence, we chose the following specific construction of $\bm{B}$:  
	\begin{align*}
	\bm B_1 & = \frac{\bm T \bm A_1 - a_m \bm A_1}{\sqrt{b_m-a_m^2}},
	\end{align*}
	where $a_m  = \bm A_1^\UH \bm T \bm A_1$ and $b_m = \bm A_1^\UH \bm T^2 \bm A_1.$
	With this choice, we note that 
	\begin{align*}
	\bm B^\UH \bm T \bm A_1 & = [\bm B_1^\UH \bm T \bm A_1, \bm B_2^\UH \bm T \bm A_1 \dots \bm B_{m-1}^\UH \bm T \bm A_1]^\UH \\ &= \sqrt{b_m - a_m^2} \bm e_1.
	\end{align*}
	Hence $\bm E(\vartheta) = \bm B^\UH \bm T \bm B + \vartheta (b_m-a_m^2)\bm e_1 \bm e_1^\UH$. To obtain the eigenvalues of $\bm E(\vartheta)$ we use its characteristic polynomial. To evaluate the characteristic polynomial of $\bm E(\vartheta)$, we connect it to the characteristic polynomial of $\bm O^\UH \bm T \bm O$, where $\bm O = [\bm A_1, \bm B]$. Note that $\bm O $ is a unitary matrix. First, we have
	\begin{align*}
	\bm O \bm ^\UH \bm T \bm O & = \begin{bmatrix} \bm A_1^\UH \bm T \bm A_1 & \bm A_1^\UH \bm T \bm B \\ \bm B^\UH \bm T \bm A_1 & \bm B^\UH \bm T \bm B
	\end{bmatrix} \\
	&= \begin{bmatrix} a_m & \sqrt{b_m - a_m^2} \bm e_1^\UH \\
	\sqrt{b_m - a_m^2} \bm e_1 & \bm B^\UH \bm T \bm B
	\end{bmatrix}.
	\end{align*}
	Consider the following matrix equation: 
	\begin{align}
	&\begin{bmatrix}
	a_m + \frac{1}{\vartheta} & \bm 0^\UH \\
	\bm 0 & \bm E(\vartheta)
	\end{bmatrix}  = \begin{bmatrix}
	a_m + \frac{1}{\vartheta} & \bm 0^\UH \\
	\bm 0 & \bm B^\UH \bm T \bm B 
	\end{bmatrix} \nonumber  \\&\hspace{5cm} +\vartheta (b_m - a_m^2) \bm e_2 \bm e_2^\UH  \nonumber \\
	 &= \begin{bmatrix}
	a_m & \sqrt{b_m - a_m^2} \bm e_1^\UH \\
	\sqrt{b_m - a_m^2} \bm e_1 & \bm B^\UH \bm T \bm B 
	\end{bmatrix}   \nonumber\\ {}  & \hspace{1.5cm} +  \begin{bmatrix} 1/\vartheta & -\sqrt{b_m-a_m^2} & \bm 0_{m-2,1}^\UH \\ -\sqrt{b_m-a_m^2} & \vartheta(b_m-a_m^2) & \bm 0^\UH_{m-2,1} \\ \bm 0_{m-2,1} & \bm 0_{m-2,1} & \bm 0_{m-2,m-2} \end{bmatrix} \nonumber \\
	 &= \bm O^\UH \bm T \bm O + \begin{bmatrix} 1/\sqrt{\vartheta} \\ - \sqrt{\vartheta(b_m-a_m^2)} \\ \bm 0_{m-2,1} \end{bmatrix} \begin{bmatrix} 1/\sqrt{\vartheta} \\ - \sqrt{\vartheta(b_m-a_m^2)} \\ \bm 0_{m-2,1} \end{bmatrix}^{\UH} \nonumber \\
	 &= \bm O^\UH (\bm T + \bm u \bm u^\UH) \bm O, \label{key_matrix_eq}
	\end{align}
	where 
	\begin{align*}
	\bm u & = \bm O \cdot \begin{bmatrix} 1/\sqrt{\vartheta} \\ - \sqrt{\vartheta(b_m-a_m^2)} \\ \bm 0_{m-2,1} \end{bmatrix} = \frac{1}{\sqrt{\vartheta}} \bm A_1 - \sqrt{\vartheta(b_m - a_m^2)} \bm B_1 \\
	& = \left( \frac{1}{\sqrt{\vartheta}} + a_m \sqrt{\vartheta} \right) \bm A_1 - \sqrt{\vartheta} \bm T \bm A_1 
	\end{align*}
	Therefore, $$|u_i|^2 =  \frac{(1+a_m\vartheta - \vartheta T_i)^2 |A_{1i}|^2}{\vartheta}.$$
	Now, we can compute the characteristic polynomial of $\bm E(\vartheta)$. We have 
	\begin{align*}
	&\det(\lambda \bm I - \bm E(\vartheta)) \\&\hspace{1cm}= \frac{1}{\lambda - a_m - \frac{1}{\vartheta}} \det \left( \lambda \bm I -  \begin{bmatrix}
	a_m + \frac{1}{\vartheta} & \bm 0^\UH \\
	\bm 0 & \bm E(\vartheta)
	\end{bmatrix} \right) \\
	&\hspace{1cm} = \frac{1}{\lambda - a_m - 1/\vartheta} \cdot \det(\lambda \bm I - \bm T - \bm u\bm u^\UH) \\
	&\hspace{1cm} = \frac{\det(\lambda I - \bm T)}{\lambda - a_m - 1/\vartheta} \cdot (1-\bm u^\UH (\lambda \bm I - \bm T)^{-1} \bm u).
	\end{align*}
	Note that
	\begin{align*}
	&1-\bm u^\UH (\lambda \bm I - \bm T)^{-1} \bm u  = 1 - \sum_{i=1}^m \frac{|u_i|^2}{\lambda - T_i}  \\
	&= 1 - \frac{1}{\vartheta}\sum_{i=1}^m \frac{(1+ a_m\vartheta - \lambda \vartheta + (\lambda - T_i)\vartheta)^2 |A_{1i}|^2}{\lambda - T_i} \\
	& = 1 - \frac{(1+ a_m\vartheta - \lambda \vartheta)^2}{\vartheta} \cdot \left(\sum_{i=1}^m \frac{|A_{1i}|^2}{\lambda - T_i} \right) \\& \hspace{1cm} - \vartheta \cdot \left( \sum_{i=1}^m (\lambda - T_i) \cdot |A_{1i}|^2 \right) - 2(1+a_m\vartheta - \lambda \vartheta) \\
	& = -(1-\lambda \vartheta + a_m\vartheta) \cdot \left(1 + \frac{1-\lambda \vartheta + a_m\vartheta}{\vartheta} Q_m(\lambda) \right),
	\end{align*}
	Where $Q_m(\lambda)$ is defined in the following way: 
	\begin{align*}
	Q_m(\lambda)  \Mydef \sum_{i=1}^m \frac{|A_{1i}|^2}{\lambda - T_i}.
	\end{align*}
	Hence,
	\begin{align} \label{eq: characteristic of F}
	&\det(\lambda \bm I - \bm E(\vartheta))  =\nonumber\\ & \hspace{2cm}  \det(\lambda \bm I - \bm T) (\vartheta + (1-\lambda \vartheta + a_m\vartheta) Q_m(\lambda)).
	\end{align}
	We emphasize that the above equation \emph{does not imply} that $T_1,T_2, \dots ,T_m$ are the eigenvalues of $\bm E(\vartheta)$. This is because while $\det(\lambda \bm I - \bm T)$ has zeros at $T_i$, the function $Q_m(\lambda)$ has poles at $T_i$. This prevents us from concluding that ${\det(\lambda \bm I - \bm E(\vartheta)) = 0}$ when $\lambda = T_i$. However, we can make the following observations: 
	\begin{enumerate}
		\item By Cauchy's interlacing theorem, we have \begin{align}
		&\lambda_1(\bm T + \vartheta (\bm T \bm A_1)(\bm T \bm A_1)^\UH) \geq T_{(1)} \nonumber \\ & \hspace{3cm}\geq \lambda_2(\bm T + \vartheta (\bm T \bm A_1)(\bm T \bm A_1)^\UH) \nonumber \\ & \hspace{3cm}\geq T_{(2)}.
		\label{interlacing_1}
		\end{align}
		The above is also true for the eigenvalues of: $$\bm O^\UH (\bm T + \vartheta (\bm T \bm A_1)(\bm T \bm A_1)^\UH  ) \bm O,$$ since $\bm O$ is a unitary matrix.
		\item \eqref{key_matrix_eq} shows that  $\bm E(\vartheta)$ is a principal submatrix of $$\bm O^\UH (\bm T + \vartheta (\bm T \bm A_1)(\bm T \bm A_1)^\UH  ) \bm O.$$ Hence, the eigenvalues of $\bm E(\vartheta)$ will interlace the eigenvalues of $\bm O^\UH (\bm T + \vartheta (\bm T \bm A_1)(\bm T \bm A_1)^\UH  ) \bm O$: 
		\begin{align}
		\lambda_1(\bm T + \vartheta (\bm T \bm A_1)(\bm T \bm A_1)^\UH &\geq \lambda_1(\bm E(\vartheta)) \nonumber \\ &\geq  \lambda_2(\bm T + \vartheta (\bm T \bm A_1)(\bm T \bm A_1)^\UH \nonumber \\&\geq \lambda_2(\bm E(\vartheta)).
		\label{interlacing_2}
		\end{align}
		Combining \eqref{interlacing_1} and \eqref{interlacing_2}, one obtains
		\begin{align*}
		\lambda_2(\bm E(\vartheta)) \leq T_{(1)}, \; \lambda_1(\bm E(\vartheta)) \geq T_{(2)}.
		\end{align*}
		This proves statement (1) in the lemma.
		This means that $\bm E(\vartheta)$ has atmost one eigenvalue bigger than $T_{(1)}$. If $\lambda_1(\bm E(\vartheta)) \leq T_{(1)}$, then it has no outlying eigenvalue, if $\lambda_1(\bm E(\vartheta)) > T_{(1)}$, it has exactly  one. We call this eigenvalue an outlying eigenvalue for reasons that will be clear later.
		\item The outlying eigenvalue of $\bm E(\vartheta)$ (if it exists) is a root of the characteristic polynomial:
		\begin{align*}
		&\det(\lambda \bm I - \bm E(\vartheta))  = \\&\hspace{1cm} \det(\lambda I - \bm T) \cdot(\vartheta + (1-\lambda \vartheta + a_m\vartheta) Q_m(\lambda)).
		\end{align*}
		Since this root lies in $(T_{(1)},\infty)$, it must be a root of: 
		\begin{align} \label{eq: Q hat root}
		Q_m(\lambda) & = \frac{1}{\lambda - a_m - 1/\vartheta}, \; \lambda > T_{(1)}.
		\end{align}
		Observing that:
		\begin{align*}
		    \lambda > T_{(1)} & \implies Q_m(\lambda) > 0, \\
		    \lambda > a_m + 1/\vartheta & \implies (\lambda - a_m - 1/\vartheta)^{-1} > 0,
		\end{align*} we conclude the outlying eigenvalue is the unique solution (if it exists) to: 
		\begin{align*}
		Q_m(\lambda) = \frac{1}{\lambda - a_m - 1/\vartheta}, \; \lambda > \max(a_m+1/\vartheta, T_{(1)}).
		\end{align*}
		This proves statement (2).
		\item Finally, we observe that $\bm E(\vartheta)$ is a positive semidefinite matrix for all $\vartheta \geq 0$, which shows $\lambda_{m-1}(\bm E(\vartheta)) \geq 0$. Also, we have $\lambda_1(\bm E(\vartheta)) \leq \| \bm E (\vartheta)\|\leq \|\bm B\|^2 \|\bm T + \vartheta \bm T \bm A_1 (\bm T \bm A_1)^\UH \|$. Note that $\|\bm B\| \leq 1$ and $\|\bm T\| \leq 1$ and $\|\bm T \bm A_1 (\bm T \bm A_1)^\UH\| = \bm A_1^\UH \bm T^2 \bm A_1 \leq T_{(1)}^2 \leq 1$. Hence, by the triangle inequality we have $\lambda_1(\bm E(\vartheta)) \leq 1 + \vartheta$. This proves statement (3) of the lemma.
	\end{enumerate}
\end{proof}

The following lemma analyzes the concentration of the function $Q_m(\lambda)$ to the deterministic function $Q(\lambda)$.

\begin{lemma} Suppose $ \frac{m}{n} = \delta $.  For a Lipschitz function $\T$  whose range is in $[0,1]$, there exists an event of probability 1, on which the following three statements hold:
	\begin{enumerate}		
		\item   $  \frac{1}{m} \sum_{i=1}^m \delta_{T_i} \explain{d}{\rightarrow} \mathcal{L}_T$,
		\item $Q_m(\lambda) \rightarrow Q(\lambda) \quad \forall \; \lambda \in (1, \infty)$,
		\item $a_m \rightarrow \E |Z|^2 T$.
	\end{enumerate}
	\label{concentration_result}
\end{lemma}
In the above equations,   $Z \sim \cgauss{0}{1}$, and $T = \T(|Z|/\sqrt{\delta})$. Furthermore, $\mathcal{L}_T$ denotes the law of the random variable $T$, and  
\begin{align*}
Q(\lambda) & = \E \left[ \frac{|Z|^2}{\lambda - T} \right].
\end{align*}
\begin{proof}
It is sufficient to show each item holds almost surely.
\begin{enumerate}
\item The argument for this part is a minor modification of the argument sketched in \cite{spruill2007asymptotic}. To prove statement (1) it suffices to show that 
\begin{equation} \label{emprical convergence of A1i}
\frac{1}{m} \sum_{i = 1}^n \delta_{\sqrt{m} {\envert{ A_{i1}}}} \xrightarrow{d} Z,
\end{equation}
almost surely.  Because if we have \eqref{emprical convergence of A1i}, then for every bounded continuous function $f$,
\begin{equation*}
f \intoo{ \T \intoo{ \sqrt{n} \envert{A_{i 1}} }} = g \intoo{\sqrt{m}  \envert{A_{1 i}}},
\end{equation*}
where $g(x) = f (\T(\frac{\envert{x}}{\sqrt{\delta}}))$ is a bounded continuous function as well.  Hence by \eqref{emprical convergence of A1i},
\begin{equation*}
\frac{1}{m} \sum_{i = 1}^m {f(T_i)} \to \e{g(Z)} = \e{ f \intoo{\T \intoo{\frac{Z}{\sqrt{\delta}}}} },
\end{equation*}
which implies
$ \frac{1}{m} \sum_{i = 1}^m \delta_{T_i} \xrightarrow{d} \mathcal{L}_T$.

To show \eqref{emprical convergence of A1i}, note that $\bf A_1$ has the same distribution as $  \frac{\bf z}{\enVert{\bf z}}$, where $ \bf z = \intoo{z_1, ..., z_m}$, and $z_i \overset{i.i.d.}{\sim}  \mathcal{CN} (0 ,1) $.  Let $\Phi$ denote the cumulative distribution function of a standard normal random variable and define 
\begin{eqnarray*}
F_m(t) &\Mydef& \frac{1}{m} \sum\limits_{i = 1}^m \mathbf{1} \intoo{\sqrt{m} \envert{A_{1i}} \leq t}, \nonumber \\
G_m(t) &\Mydef& \frac{1}{m} \sum\limits_{i = 1}^m \mathbf{1} \intoo{z_i \leq t}. \nonumber
\end{eqnarray*}
 Then, we have
\begin{equation} \label{eq: F_m same dist as G_m}
F_m(t) \overset{d}{=} G_m \intoo{t \frac{\enVert{\bf z}}{\sqrt{m}}}.
\end{equation}
Moreover,
\begin{align*} \label{eq: G_m to phi}
&G_m \intoo{t \frac{\enVert{\bf z}}{\sqrt{m}}} - \Phi(t) = \\& \hspace{0.8cm} G_m \intoo{t \frac{\enVert{\bf z}}{\sqrt{m}}} - \Phi \intoo{t \frac{\enVert{\bf z}}{\sqrt{m}}} + \Phi \intoo{t \frac{\enVert{\bf z}}{\sqrt{m}}} - \Phi(t)
\\& \hspace{1cm}\xrightarrow{a.s.} 0 + 0.
\end{align*}
$G_m(t \enVert{\bf z}) - \Phi(t \enVert{\bf z})$ goes to $0$ almost surely by Glivenko-Cantelli lemma. Furthermore, since $$\frac{\enVert{\bf z}}{\sqrt{m}} \xrightarrow{a.s.} 1, $$ and $\Phi$ is a continuous function we conclude that $$\Phi \intoo{t \frac{\enVert{\bf z}}{\sqrt{m}}} - \Phi(t) \overset{a.s.}{\rightarrow} 0.$$  Hence,
\begin{equation*}
F_m(t) \to \Phi(t),
\end{equation*}
almost surely which yields \eqref{emprical convergence of A1i}.
\item We now focus on the proof of statement (2). Let $$\mathcal{C}_k \Mydef \intcc{1 + \frac{1}{k}, k}, \quad k \in \mathbb{N}.$$  We will show that
\begin{equation} \label{convergence in C_k}
Q_m(\lambda) \to Q(\lambda) \quad \forall \; \lambda \; \in \; \mathcal{C}_k, 
\end{equation}
  almost surely. This means there is a set $\mathcal{C}_k^\prime$, with measure $0$, out of which we have the convergence for all $ \lambda \in \mathcal{C}_k$. If we define $ \mathcal{C}^\prime \Mydef \bigcup \limits_{k = 1}^\infty \mathcal{C}_k^\prime$, then $Q_m(\lambda) \to Q(\lambda) \quad \forall \lambda \in \intoo{1, \infty}$ out of $\mathcal{C'}$ and clearly $ \p{\mathcal{C'}} = 0 $.

First note that $ \bf A_1 \overset{d}{=} \frac{\bf z}{\enVert{\bf z}} $, where $$ \bf z = \intoo{z_1, ..., z_m}, \quad z_i \overset{i.i.d.}{\sim} \mathcal{CN} (0 ,1). $$  Define
\begin{equation} \label{def: tild Q}
\tilde{Q}_m(\lambda) \Mydef \frac{1}{m} \sum_{i = 1}^m \frac{\envert{z_i}^2}{\lambda - \T \intoo{\frac{\envert{z_i}}{\sqrt{\delta}}}}.
\end{equation}

Note that  for a fixed $\lambda$ we have $ \tilde{Q}_m(\lambda) \to Q(\lambda) $ almost surely by the strong law of large numbers.  Since $ \tilde{Q}_m (\lambda) $ is a decreasing function in $\lambda$ and we have $ \tilde{Q}_m (\lambda) \to Q(\lambda) \quad \forall \lambda \in \mathcal{C}_k \cap \mathbb{Q} $  almost surely, we obtain $\tilde{Q}_m (\lambda) \to Q(\lambda)$ for all $ \lambda \in \mathcal{C}_k$ with probability $1$. Hence, it suffices to show under an event that holds with probability 1,
\begin{equation} \label{Q - tilde Q to 0}
Q_m(\lambda) - \tilde{Q}_m (\lambda)  \to 0 \quad \forall \lambda \in \mathcal{C}_k.
\end{equation}

To prove \eqref{Q - tilde Q to 0}, we will find a sequence $\tau_m$ such that $\tau_m \rightarrow 0$ as $m \rightarrow \infty$, and,
\begin{align*}
    \sum_{m \geq 1} \p{\sup_{\lambda \in \mathcal{C}_k} \envert{ Q_m(\lambda) - \tilde{Q}_m (\lambda) } > \tau_m} < \infty.
\end{align*} With this, Borel-Cantelli lemma yields that event $$E = \cbr{ \sup\limits_{\lambda \in \mathcal{C}_k} \envert{Q_m(\lambda) - \tilde{Q}_m (\lambda)} > \tau_m \; \text{infinitely often}  }$$ has measure $0$.  Out of the event $E$ we have \eqref{Q - tilde Q to 0} as it was desired.

Define the events:
\begin{align*}
    E_1 &\triangleq \cbr{ \sup\limits_{i \leq m } \envert{z_i} \leq \sqrt{ 6 \log m} }, \\
    E_{2, \epsilon} &\triangleq \cbr{ \envert{\frac{\enVert{\bf z}^2}{m} - 1} \leq \epsilon },
\end{align*}  where $\epsilon$ is parameter we will set later. Note that,
\begin{align*}
    &\envert{ Q_m(\lambda) - \tilde{Q}_m (\lambda) }  \leq \\ & \hspace{1.5cm}
\sum_{i = 1}^m \frac{\envert{z_i}^2}{\enVert{\bf z}^2} \envert{\frac{\frac{\enVert{\bf z}^2}{m}}{\lambda - \T \intoo{\frac{\envert{z_i}}{\sqrt{\delta}}}} - \frac{1}{\lambda - \T \intoo{\frac{\sqrt{n}}{\enVert{\bf z}} \envert{z_i} }}} \\&  \hspace{1.5cm} \leq\mathsf{I} + \mathsf{II},
\end{align*}
where we defined the terms $\mathsf{I}, \mathsf{II}$ as:
\begin{align*}
    \mathsf{I} & = \left| \frac{\|\bm z\|^2}{m}  - 1\right| \cdot  \sum_{i=1}^m \frac{|z_i|^2}{\|\bm z\|^2}  \cdot \left| \frac{1}{\lambda - \T \big( \frac{|z_i|}{\sqrt{\delta}}\big)} \right| \\
    \mathsf{II} & = \sum_{i=1}^m \frac{|z_i|^2}{\|\bm z\|^2} \cdot  \frac{\left|\T \big( \frac{|z_i|}{\sqrt{\delta}}\big) -  \T \big( \frac{\sqrt{n} |z_i|}{\|\bm z\|}\big)\right|}{ \left|\lambda - \T \big( \frac{|z_i|}{\sqrt{\delta}}\big) \right| \cdot \left| \lambda - \T \big( \frac{\sqrt{n} |z_i|}{\|\bm z\|}\big)\right|} . 
\end{align*}
Using the fact that $\bm z \in E_1 \cap E_{2,\epsilon}$ and $\lambda \in \mathcal{C}_k$, we have,
\begin{align*}
    \mathsf{I} & \leq k\epsilon  , \\
    \mathsf{II} & \leq k^2 \cdot \max_{i \leq n} \left|\T \bigg( \frac{|z_i|}{\sqrt{\delta}}\bigg) -  \T \big( \frac{\sqrt{n} |z_i|}{\|\bm z\|}\big)\right|.
\end{align*}
Observe that, on the event $E_1 \cap E_{2,\epsilon}$,
\begin{align*}
    \envert{ \frac{\envert{z_i}}{\sqrt{\delta}} - \frac{\sqrt{n}}{\enVert{ \bf z}} \envert{z_i}} & \leq \frac{\envert{z_i}}{\sqrt{\delta}} \envert{ 1 - \frac{\sqrt{m}}{\enVert{ \bf z}} } \\
    & \leq \sqrt{6\log(m)} \cdot  \envert{ 1 - \frac{\sqrt{m}}{\enVert{ \bf z}} } \\
    & \leq \sqrt{6\log(m)} \cdot \envert{ 1 - \frac{{m}}{\enVert{ \bf z}^2} }  \\
& \leq \sqrt{6\log(m)} \cdot \frac{\epsilon}{1 - \epsilon}.
\end{align*}
Since $\T$ was assumed to be Lipchitz, 
\begin{align*}
     \mathsf{II} & \leq k^2 \cdot \max_{i \leq n} \left|\T \bigg( \frac{|z_i|}{\sqrt{\delta}}\bigg) -  \T \big( \frac{\sqrt{n} |z_i|}{\|\bm z\|}\big)\right| \\
     & \leq k^2 \cdot \|\T\|_\mathsf{Lip} \cdot \sqrt{6\log(m)} \cdot \frac{\epsilon}{1-\epsilon},
\end{align*}
where $\|\T\|_{\mathsf{Lip}}$ denotes the Lipchitz constant of $\T$.
Hence, when $m \geq e^2$, setting $\epsilon  = \frac{1}{\log(m)} \leq 0.5$, we obtain, on the event $E_1 \cap E_{2,\epsilon}$
\begin{align} \label{Q_m - Q less than epsilon log}
\envert{ Q_m(\lambda) - \tilde{Q}_m (\lambda) } & \leq \tau_m, \; \forall \; \lambda \; \in \; \mathcal{C}_k.
\end{align}
where
\begin{align*}
     \tau_m & = \frac{k}{\log(m)} + \frac{2k^2 \cdot \|\T \|_{\mathsf{Lip}}}{\sqrt{\log(m)}}.
\end{align*}
Note that $\tau_m \rightarrow 0$ as $m \rightarrow \infty$ as required. And, 
\begin{align*}
    &\p{\sup_{\lambda \in \mathcal{C}_k} \envert{ Q_m(\lambda) - \tilde{Q}_m (\lambda) } > \tau_m}  \\&\hspace{3cm}\leq \p{{E}_1^c} + \p{E_{2,\epsilon}^c} \\
    & \hspace{3cm} \leq  2 \cdot m^{-2} + 2 e^{-\frac{m}{8\log^2(m)}}, 
\end{align*}
where the last step follows from standard bounds on the tail Gaussian random variables and $\chi^2$ random variables. 
In particular, we have,
\begin{align*}
    \sum_{m \geq 1} \p{\sup_{\lambda \in \mathcal{C}_k} \envert{ Q_m(\lambda) - \tilde{Q}_m (\lambda) } > \tau_m} < \infty,
\end{align*}
as required. 
\item The proof is similar to the proof of the second statement. Hence, we skip the details. Note that if we define
\begin{equation*}
W_m = \sum_{i = 1}^n \envert{A_{1 i}}^2 \T(\envert{A_{1 i}} \sqrt{n}),
\end{equation*}
then it again converges under the event $E_1 \cap E_{2, \epsilon}$, defined in the proof of statement (2).
\end{enumerate}
\end{proof}
The next lemma analyzes the properties of the limiting fixed point equation $Q(\lambda) = (\lambda - \E |Z|^2 T - 1/\vartheta)^{-1}$. 	Define the critical value $\vartheta_c$ as: 
	\begin{align*}
	\vartheta_c & \Mydef \left(1 - \left( \E \left[ \frac{|Z|^2}{1-T} \right] \right)^{-1} - \E[|Z|^2 T]  \right)^{-1} \geq 0.
	\end{align*}

\begin{lemma} \label{Q_eq_population}
	Consider the fixed point equation (in $\lambda$)
	\begin{align}\label{eq:fixedpoint_as}
	\lambda - \E[|Z|^2 T] - 1/\vartheta & = \frac{1}{\E \left[ \frac{|Z|^2}{\lambda - T} \right]},
	\end{align}
	on the domain:
	\begin{align*}
	    \lambda >  \max(1, \E[ |Z|^2 T]+1/\vartheta).
	\end{align*}
We have
	\begin{enumerate}
		\item If $\vartheta > \vartheta_c$, then the above equation has exactly 1 solution, denoted by $\lambda = \theta(\vartheta)$. Furthermore, 
		\begin{align*}
		&\lambda - \E[|Z|^2 T] - 1/\vartheta  >  \frac{1}{\E \left[ \frac{|Z|^2}{\lambda - T} \right]} \\  & \hspace{2cm} \forall \; \lambda \in \left( \max(1, \E[ |Z|^2 T]+1/\vartheta), \theta(\vartheta) 
		\right), \\
		&\lambda - \E[|Z|^2 T] - 1/\vartheta  < \frac{1}{\E \left[ \frac{|Z|^2}{\lambda - T} \right]} \; \forall \; \lambda \in \left(\theta(\vartheta),\infty  \right).
		\end{align*}
		Furthermore, we have $\theta(\vartheta)$ is an increasing function of $\vartheta$ and $\lim_{\vartheta \rightarrow \infty} \theta(\vartheta) = \infty$.
		\item If $\vartheta \leq \vartheta_c$, then the equation has no solutions. For any $\vartheta \leq \vartheta_c$, we define $\theta(\vartheta) = 1$.
	\end{enumerate}

\end{lemma}
\begin{proof}
	The following change of measure simplifies some of the proofs: 
	\begin{align*}
	p(z) & \Mydef \frac{|z|^2}{\pi} \exp(-|z|^2), \\
	\tE[f(Z)] & \Mydef \int f(z) p(z) \diff z.
	\end{align*}
	Note that $p(z)$ is a proper probability density function since $\int p(z) \diff z = \E[|Z|^2] = 1$.  With this notation, \eqref{eq:fixedpoint_as} can be written as
	\begin{align*}
	\lambda - \tE[T] - 1/\vartheta & = \frac{1}{\tE \left[ \frac{1}{\lambda - T} \right]}, \; \lambda > \max(1, \tE[ T]+1/\vartheta).
	\end{align*}
	Define the random variable $G(\lambda)  = (\lambda-T)^{-1}$. Note that $G^\prime(\lambda) = - G^2(\lambda)$. Further, define
	\begin{align*}
	f(\lambda) & \Mydef \frac{1}{\tE \left[ G(\lambda) \right]}; \; \lambda \in [1,\infty).
	\end{align*}
	The first two derivatives of $f(\lambda)$ are
	\begin{align*}
	f^\prime(\lambda) &= \frac{\tE[G^2]}{\tE[G]^2}, \nonumber \\
	 f^{\prime \prime}(\lambda) &= -2 \cdot  \frac{\tE[G^3]\tE[G] - \tE[G^2]^2 }{\tE[G]^3}. 
	\end{align*}
	First, since $f^\prime(\lambda) \geq 0$, the function $f(\lambda)$ is increasing. By Jensen's Inequality $f^\prime(\lambda) \geq 1$. Since the equality holds if and only if $G$ is deterministic, and we have assumed that the support of $T$ is  $[0,1]$, we conclude that $f(\lambda)>1$. Noting that $G \geq 0$ and applying Chebychev's association inequality (See Fact \ref{chebychev_association}, Appendix \ref{analytic_appendix}) with $B=A=G$ and $f(a) = g(a) = a$ gives ${f}^{\prime\prime}(\lambda) \leq 0$. Hence $f(\lambda)$ is an increasing, concave function and ${f}^\prime(\lambda) > 1$. 
	
	Next, we claim that $f(\lambda) = \lambda - \tE[ T] - 1/\vartheta$ can have atmost one solution in $(1,\infty)$. To see this, let $\lambda_1$ be the first point at which the two curves intersect. Hence $f(\lambda_1) = \lambda_1 - \tE[T] - 1/\vartheta$. Furthermore $$f^\prime(\lambda) > 1 = \frac{\diff (\lambda - \tE[T] - 1/\vartheta)}{ \diff \lambda}. $$
	Hence there can be no other intersection point of the two curves after $\lambda_1$.
	
	Now consider the following two cases:
	
		\textit{Case 1: $\vartheta >  \vartheta_c$.}
		First note that since $(1-x)^{-1}$ is a convex function on $(-\infty,1]$, according to Jensen's Inequality
		\begin{align*}
		\tE \left[ \frac{1}{1-T} \right] \geq \frac{1}{1-\tE[T]} \geq 0.
		\end{align*}
		Hence,
		\begin{align*}
		\frac{1}{\vartheta_c}  = 1- \left( \tE \left[ \frac{1}{1-T} \right]  \right)^{-1} - \tE[T] \geq 0.
		\end{align*}
		This shows that $\vartheta_c \geq 0$. Furthermore, 
		\begin{align*}
		\vartheta > \vartheta_c \Longleftrightarrow (\lambda - \tE[T] - 1/\vartheta)_{\lambda = 1} > f(1).
		\end{align*}
		On the other hand, we can also compare the limiting behavior of $\lambda - \tE[T] - 1/\vartheta$ and $f(\lambda)$ as $\lambda \rightarrow \infty$. We have
		\begin{align*}
		\frac{\lambda - \tE[T] - 1/\vartheta}{ \lambda } & = 1 - \frac{\tE[T] + 1/\vartheta}{\lambda}, 
		\end{align*}
		and
		\begin{align*}
		\frac{f(\lambda)}{\lambda} &= \frac{1}{\tE \left[ \frac{1}{1 - T/\lambda} \right]} = \left( \tE \left[ \sum_{n=0}^\infty \left(\frac{T}{\lambda} \right)^n \right] \right)^{-1} \\
		& = \left(1 + \tE[T]/\lambda + o(1/\lambda) \right)^{-1}\\ &= 1 - \frac{\tE[T]}{\lambda} + o(\lambda^{-1}).
		\end{align*}
		Hence, $f(\lambda) > \lambda - \tE[T]-1/\vartheta$ for $\lambda$ large enough and $f(1) < 1 - \tE[T]-1/\vartheta$. Hence the functions $f(\lambda)$ and $1-\tE[T]-1/\vartheta$ intersect once in $(1,\infty)$. Finally note that,
		\begin{align*}
		\frac{1}{\vartheta} + \tE[T] &< \frac{1}{\vartheta_c} + \tE[T] = 1- \left( \tE \left[ \frac{1}{1-T} \right]  \right)^{-1} \\ &\leq 1.
		\end{align*}
		Hence $f(\lambda) = \lambda - \tE[T]-1/\vartheta$ has exactly one solution in $\lambda \geq \max(1,\tE[T]+1/\vartheta)$ as claimed. By the Implicit Function Theorem, we can compute
		\begin{align}
		\theta^\prime(\vartheta) & = \frac{1/\vartheta^2}{f^\prime(\theta(\vartheta))-1} \geq 0.
		\end{align}
		Hence $\theta(\vartheta)$ is an increasing function of $\vartheta$. Finally, we verify that $\lim_{\vartheta \rightarrow \infty} \theta(\vartheta) = \infty$. Suppose that this is not the case, i.e. $\theta(\vartheta) \rightarrow \theta_\infty < \infty$ as $\vartheta \rightarrow \infty$. Recalling the fixed point characterization of $\theta(\vartheta)$, we obtain that $\theta_\infty$ satisfies the fixed point equation
		\begin{align*}
		\theta_\infty - \tE[T]  = \frac{1}{\tE\left[ \frac{1}{\theta_\infty - T}\right]}.
		\end{align*}
		This means that Jensen's Inequality applied to the strictly convex function $(\theta_\infty - t)^{-1}$ should be tight. This means under the tilted measure ($\tE$), $T$ is deterministic. This is not possible since we have assumed that $T$ is supported on $[0,1]$. 
		
		\textit{Case 2: $\vartheta \leq \vartheta_c$} As in Case 1 we argue (this time with the opposite conclusion) that $$\vartheta \leq \vartheta_c \implies f(1) \geq (\lambda-\tE[T]-1/\vartheta)_{\lambda = 1}$$
	Furthermore, since ${f}^\prime(\lambda) > \frac{\diff (\lambda - \tE[T]-1/\vartheta)}{\diff \lambda} = 1,$
	$f(\lambda) = \lambda - \tE[T] - 1/\vartheta$ has no solution in $(1,\infty)$.
\end{proof}
Combining the above sequence of lemmas, we obtain the following proposition about the spectrum of the matrix $\bm E(\vartheta)$.
\begin{proposition} Let $\bm E(\vartheta)  = \bm B^\UH (\bm T + \vartheta \bm T \bm A_1 (\bm T \bm A_1)^\UH)) \bm B$. Then, there exists an event of probability 1, on which we have,
	\begin{enumerate}
		\item $\mu_{\bm E(\vartheta)} \explain{d}{\rightarrow} \mathcal{L}_T$.
		\item If $\vartheta \leq \vartheta_c$, $\sigma(\bm E(\vartheta)) \subset [0,1]$.
		\item If $\vartheta > \vartheta_c$, then $\lambda_i(\bm E(\vartheta)) \in [0,1] \; \forall \; i \geq 2$, and,
		\begin{align*}
		\lambda_1(\bm E(\vartheta)) \explain{a.s.}{\rightarrow} \theta(\vartheta),
		\end{align*}
		where $\theta(\vartheta)$ is the unique solution to the equation (in $\lambda$): 
		\begin{align*}
		\lambda - \E[|Z|^2 T] - 1/\vartheta & = \frac{1}{\E \left[ \frac{|Z|^2}{\lambda - T} \right]}, 
		\end{align*}
		in the domain:
		\begin{align*}
		    \lambda >  \max(1, \E[ |Z|^2 T]+1/\vartheta).
		\end{align*}
	\end{enumerate}
	\label{E_spectrum_proposition}
\end{proposition}
\begin{proof} 
	We restrict ourselves to the event guaranteed by Lemma \ref{concentration_result}, on which,
	\begin{enumerate}
		\item $a_m \rightarrow \E |Z|^2 T$
		\item $\frac{1}{m} \sum_{i=1}^m \delta_{T_i} \explain{d}{\rightarrow} \mathcal{L}_T$
		\item $Q_m(\lambda) \rightarrow Q(\lambda) \; \forall \; \lambda \in (1, \infty)$.
	\end{enumerate}
	Let us denote this event by $\mathcal{E}$.
	Define the sequence of (random) functions $f_m(\lambda)$ as: 
	\begin{align*}
	f_m(\lambda) & = \lambda - a_m - 1/\vartheta - \left( \sum_{i=1}^m \frac{|A_{1i}|^2}{\lambda - T_i}\right)^{-1}, 
	\end{align*}
	with the domain:
	\begin{align*}
	  \lambda &> \max(1, a_m + 1/\vartheta).
	\end{align*}
	Define the (deterministic) function $f(\lambda)$: 
	\begin{align*}
	f(\lambda) & = \lambda - \E[|Z|^2 T] - 1/\vartheta-  \left(\E \left[ \frac{|Z|^2}{\lambda - T} \right]\right)^{-1}, 
	\end{align*}
	with the domain:
	\begin{align*}
	    \lambda &>  \max(1, \E[ |Z|^2 T]+1/\vartheta).
	\end{align*}
	Note that on $\mathcal{E}$, we have $f_m(\lambda) \rightarrow f(\lambda) \; \forall \; \lambda >1$.
	
	\begin{enumerate}
		\item By Lemma \ref{spectrum_E_empirical}, we know that the eigenvalues of $\bm E(\vartheta)$ interlace with the eigenvalues of the diagonal matrix $\bm T$. On the event $\mathcal{E}$, $\mu_{\bm T} \rightarrow \mathcal{L}_T$. Hence indeed $\mu_{\bm E(\vartheta)} \explain{d}{\rightarrow} \mathcal{L}_T$. This proves statement (1) of the proposition.
		\item Consider the case $\vartheta \leq \vartheta_c$. By Lemma \ref{spectrum_E_empirical}, we already know that $\lambda_2(\bm E(\vartheta)) \leq T_{(1)} \leq 1$ and $\lambda_{m-1}(\bm E(\vartheta)) \geq 0$. Hence to prove (2), it is sufficient to show that 
		$$\bar{\lambda}_1 \Mydef \underset{m \rightarrow \infty}{\lim \sup} \    \lambda_1(\bm E(\vartheta)) \leq 1, \; \text{on} \  \mathcal{E}.$$
		For the sake of contradiction, suppose that there is a realization in $\mathcal{E}$ such that $\bar{\lambda}_1> 1$.  On this realization we consider a subsequence such that $\lambda_1(\bm E(\vartheta)) \rightarrow \bar{\lambda}_1$. All the analysis henceforth is along this subsequence.  Since for all $m$ large enough $\lambda_1(\bm E(\vartheta))>1$, by Lemma \ref{spectrum_E_empirical}, we must have $f_m(\lambda_1(\bm E(\vartheta))=0$. Applying Lemma 3 from \cite{Lu17} (Appendix E), we obtain
		\begin{align*}
		0 & = f_m(\lambda_1(\bm E(\vartheta)) \rightarrow f(\bar{\lambda}_1).
		\end{align*}
		Since $\vartheta \leq \vartheta_c$, we know by Lemma \ref{Q_eq_population} that $f(\lambda) = 0$ does not have any solution in $\lambda > \max(1,\E[|Z|^2 T] + 1/\vartheta)$. Hence,
		\begin{align*}
		1 < \bar{\lambda}_1 \leq \E[|Z|^2 T] + 1/\vartheta.
		\end{align*}
		However,
		\begin{align*}
		f(\bar{\lambda}_1) & = \underbrace{\bar{\lambda}_1 - \E[|Z|^2 T] - 1/\vartheta}_{\leq 0}-  \bigg(\underbrace{\E \left[ \frac{|Z|^2}{\lambda - T} \right]}_{> 0}\bigg)^{-1} \\&< 0.
		\end{align*}
		This contradicts $f(\bar{\lambda}_1) = 0$. Hence, $\underset{m \rightarrow \infty}{\lim \sup} \ \lambda_1(\bm E(\vartheta)) \leq 1, \; \text{on $\mathcal{E}$.}$ This concludes the proof of statement (2).
		\item Now consider the case $\vartheta > \vartheta_c$. Again by Lemma \ref{spectrum_E_empirical}, we know $\lambda_i(\bm E(\vartheta)) \in [0,1]$ for all $i \geq 2$. By Lemma \ref{Q_eq_population}, we know that $f(\lambda) = 0$ has a unique solution in $\lambda> \max(1, \E|Z|^2 T  + 1/\vartheta)$ denoted by $\theta(\vartheta)$. Fix an $\epsilon$ small enough such that $[\theta(\vartheta)-\epsilon, \theta(\vartheta)+\epsilon]$ lies in the domain of $f(\lambda)$. Note that $f(\theta(\vartheta)) = 0$, while $f(\theta(\vartheta) - \epsilon)>0$ and $f(\theta(\vartheta) + \epsilon)<0$ (by Lemma \ref{Q_eq_population}).
		
		Since $a_m \rightarrow \E|Z|^2 T$, for all $m$ large enough, $[\theta(\vartheta) - \epsilon, \theta(\vartheta) + \epsilon]$ also lies in the domain of $f_m(\lambda)$. By Lemma \ref{concentration_result}, we have $f_m(\lambda) \rightarrow f(\lambda)$ for all $\lambda \in [\theta(\vartheta) - \epsilon, \theta(\vartheta) + \epsilon] $. In particular, we have, for all $n$ large enough $f_m(\theta(\vartheta)-\epsilon)>0$ while $f_m(\theta(\vartheta)+\epsilon)<0$. Hence, by Lemma \ref{spectrum_E_empirical}, we have $\lambda_1(\bm E(\vartheta)) \in [\theta(\vartheta) - \epsilon, \theta(\vartheta) + \epsilon]$ for all $n$ large enough. Hence indeed, $\lambda_1(\bm E(\vartheta)) \explain{a.s.}{\rightarrow} \theta(\vartheta)$. This proves (3).
	\end{enumerate}
\end{proof}
%
\subsection{Analysis of the Support of $\gamma \boxtimes \mathcal{L}_T$}\label{Sec:support}

	We recall that $\mathcal{L}_T$ is the law of the random variable $T = \T(|Z|/\sqrt{\delta})$, and $\gamma = \frac{1}{\delta} \delta_1 + \left( 1 - \frac{1}{\delta} \right) \delta_0$. To keep the notation clean, we will refer to the analytic transforms corresponding to the measure $\mathcal{L}_T$ with the subscript $T$, for example the Cauchy transform for the measure $\mathcal{L}_T$ will be referred to as $G_T$.We begin by computing the Cauchy Transform of $\gamma \boxtimes T$.
	\begin{lemma} Let $z \in \mathbb C^{-}$. Then, we have,
		\begin{align*}
		G_{\gamma \boxtimes T}(z) & = \frac{1}{z}\cdot \frac{1-1/\delta}{1-z w_T(1/z)}. 
		\end{align*}
		In the above display, the subordination function, $w_T(1/z)$, is the unique solution in $\mathbb C^+$ to the equation $\Lambda(1/w) = z$, where the function $\Lambda$ is defined as:
		\begin{align*}
		\Lambda(\tau) & \Mydef  \tau - \frac{(1-1/\delta)}{\E\left[\frac{1}{\tau- T}\right]}.
		\end{align*}
		\label{formula_for_cauchy_and_subordination}
	\end{lemma}
	\begin{proof}
		First we can compute the moment generating functions: 
		\begin{align*}
		\psi_{\gamma}(z) &= \frac{1}{\delta} \cdot \frac{ z}{1-z}, \\
		 \psi_{T}(z) &= -1 + \E\left[ \frac{1}{1-zT} \right].
		\end{align*}
		The $\eta$-transforms of the two measures are given by,
		\begin{align*}
		\eta_\gamma(z) &= \frac{z/\delta}{  z/\delta - z +1}, \\
		 \eta_T(z) &= \frac{\E\left[ \frac{zT}{1-zT} \right]}{\E\left[ \frac{1}{1-zT} \right]}.
		\end{align*}
		Hence, we can compute the function $Q_z$, given in Definition \ref{def_subordination},
		\begin{align*}
		Q_z(w) & = \frac{1/\delta}{(1/\delta - 1) \frac{\E\left[\frac{T}{1-w T}\right]}{\E\left[\frac{1}{1-w T}\right]} + 1/z}.
		\end{align*}
		Hence $w_T$ is the unique solution in $\mathbb C^{+}$ of the equation $Q_z(w) = w$. This equation can be simplified to
		\begin{align*}
		\frac{1}{z} & = \Lambda(1/w),
		\end{align*}
		where the function $\Lambda$ is defined as $\Lambda(\tau)  \Mydef  \tau - \frac{(1-1/\delta)}{\E\left[\frac{1}{\tau- T}\right]}.$ Hence, we can compute the moment generating function of $\gamma \boxtimes T$ in the following way:
		\begin{align*}
		\psi_{\gamma \boxtimes T}(z) & = \psi_T(w_T(z))\\ &= -1 + \E \left[\frac{1}{1-w_T(z) T} \right] \\ & \explain{(a)}{=} -1 + \frac{1-1/\delta}{1-w_T(z)/z}.
		\end{align*}
		In the above display, in the step marked (a), we used the fact that $w_T$ solves $\Lambda(1/w) = 1/z$. Finally, the Cauchy Transform of $\gamma \boxtimes T$ is given by
		\begin{align*}
		G_{\gamma \boxtimes T}(z) & = \frac{1}{z} \left( \psi_{\gamma\boxtimes T} \left( \frac{1}{z}\right) + 1 \right) \\
		& = \frac{1}{z}\cdot \frac{1-1/\delta}{1-z w_T(1/z)}. 
		\end{align*}
	\end{proof}
	Our next goal is to characterize $\text{Supp}(\gamma \boxtimes T)$. Theorem \ref{regularity_theorem} gives a complete characterization of the support of the singular part of $\gamma \boxtimes T$. Hence, we now need to understand the support of the absolutely continuous part of $\gamma \boxtimes T$. According to the Stieltjes Inversion theorem, (Theorem \ref{stieltjes_theorem}) the density of the continuous part is given by
	\begin{align*}
	&\frac{\diff (\gamma \boxtimes T)_{ac}}{\diff x}(x)  = \frac{1}{\pi} \lim_{\epsilon \rightarrow 0^{+}}\Im \;  G_{\gamma \boxtimes T}(x-i \epsilon) \\
	&\hspace{2cm}= \frac{1}{\pi x}  \Im \left( \frac{1-\frac{1}{\delta}}{1-x  \lim_{\epsilon \rightarrow 0^{+}} w_T(1/(x-i \epsilon))} \right).
	\end{align*}
	Since $\tau_T(x-i \epsilon) \Mydef 1/w_T(1/(x-i \epsilon))$ uniquely solves $\Lambda(\tau) = x-i \epsilon$ in $\mathbb C^-$, our interest will be to study the solutions of this equation for $\epsilon \approx 0$. Hence, we begin by studying the solutions of $\Lambda(\tau) = x$. Before doing so, we clarify the definition of $\Lambda(\tau)$ at $\tau = 1$ which is a subtle case because $1 \in \text{Supp}(T)$. We note that the random variable $(1-T)^{-1}$ is non-negative and hence the expectation $\E[(1-T)^{-1}]$ is well defined but might be $\infty$. If it is finite, then $\Lambda(\tau)$ is well defined at $\tau = 1$. If the expectation is $\infty$, we define $\Lambda(1) = 1$ which is consistent with intepreting $1/\infty = 0$. $\Lambda(\tau)$ is defined at $\tau = 0$ analogously. This definition ensures $\Lambda(\tau)$ is a continuous function on $(-\infty,0] \cup [1,\infty)$. Next we discuss the solutions of $\Lambda(\tau) = x$. Figure \ref{Lambda_eq_on_reals_fig} shows a typical plot $\Lambda(\tau)$. As is clear from this figure we expect the following two quantities to play major roles in determining the existence of a solution of $\Lambda(\tau) = x$:
	Define
			\begin{align*}
			\lambda_{l} & = \max_{\tau \in (-\infty,0]} \Lambda(\tau),  \; \tau_l = \argmax_{\tau \in (-\infty,0]} \Lambda(\tau)\\
			\lambda_{r} & = \min_{\tau \in [1,\infty)} \Lambda(\tau), \; \tau_{r}  = \argmin_{\tau \in [1,\infty)} \Lambda(\tau). 
			\end{align*}
Our next lemma proves the properties of $\Lambda(\tau)$ suggested by Figure \ref{Lambda_eq_on_reals_fig}.

\begin{figure*}[hbpt]
\centering
\includegraphics[width=.6\textwidth]{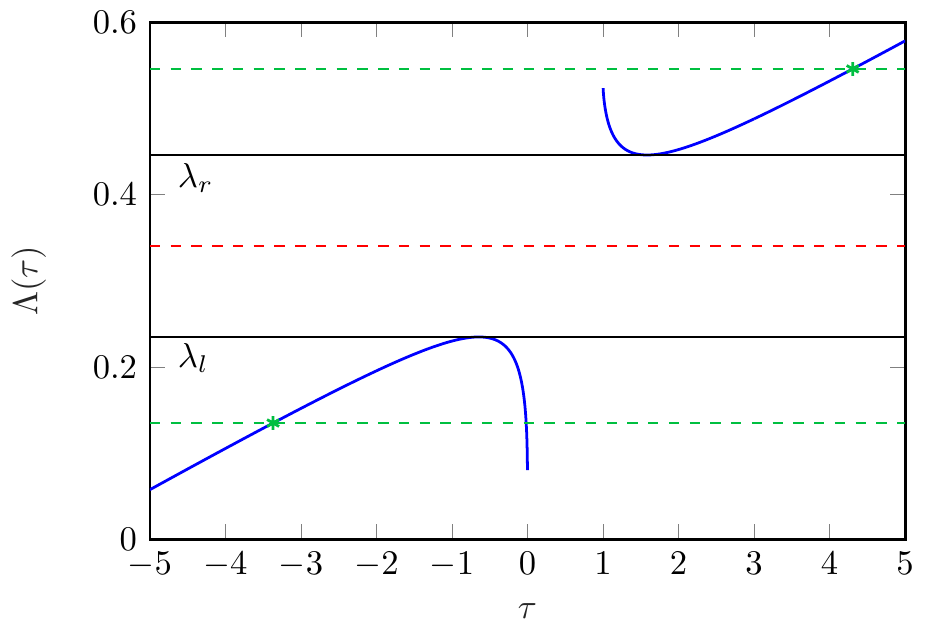}
\caption{An Illustrative plot of the function $\Lambda(\tau)$: When $\lambda_l<x<\lambda_r$, the equation $\Lambda(\tau) = x$ has no solutions. When $x\geq \lambda_r$, the equation $\Lambda(\tau) = x, \Lambda^\prime(\tau)>0$ has a unique solution in $[1,\infty)$. When $x< \lambda_l$, then $\Lambda(\tau) = x, \Lambda^\prime(\tau)>0$ has a unique solution in $(-\infty,0]$. }
		\label{Lambda_eq_on_reals_fig}
\end{figure*}


	\begin{lemma} 
The following statements are true about $\Lambda(\tau)$:	
		\begin{enumerate}
			\item $\Lambda(\tau)$ is a convex function on $[1,\infty)$ and a concave function on $(-\infty,0]$.
			\item $\lim_{\tau \rightarrow \infty} \Lambda(\tau) = \infty, \; \lim_{\tau \rightarrow -\infty} \Lambda(\tau) = -\infty$.
			\item $\lambda_r> \lambda_l \geq 0$. 
			\item Consider the 3 mutually exclusive and exhaustive cases:
			\begin{description}
				\item{Case A: $x \leq \lambda_{l}$.} There is at least one and at most two solutions to $\Lambda(\tau) = x$. All solutions lie in $(-\infty,0]$. Furthermore, when $x < \lambda_{l}$, there is exactly one solution for the equation $\Lambda(\tau) = x, {\Lambda}^\prime(\tau) > 0$. This unique solution additionally satisfies $\tau < \tau_l \leq 0$.
				\item{Case B: $\lambda_{l} < x < \lambda_{r}$.} There are no solutions of the equation $\Lambda(\tau) = x, \; \tau \in (-\infty,0] \cup [1,\infty)$.
				\item {Case C: $x \geq \lambda_{r}$.} There is at least one and at most two solutions to $\Lambda(\tau) = x$. All solutions lie in $[1,\infty)$. Furthermore, when, $x > \lambda_{r}$, there is a unique solution to $\Lambda(\tau)=x, {\Lambda}^\prime(\tau)>0$. This  solution additionally satisfies $\tau > \tau_r \geq 1$.
			\end{description}
		\end{enumerate}
		\label{solutions_real_case}
	\end{lemma}

	\begin{proof}
		\begin{enumerate}
			\item We define the random variable $G(\tau)$, 
			\begin{align*}
			G(\tau) & \Mydef \frac{1}{\tau- T}.
			\end{align*} 
			We observe that for any $\tau \in [1,\infty)$, $G(\tau) \geq 0$ where as for $\tau \in (-\infty,0]$, $G(\tau) \leq 0$. It is straightforward to see that $G^\prime(\tau) = - G^2(\tau) \leq 0.$ For notational simplicity, we will often short hand $G(\tau)$ as $G$. We have
			\begin{align*}
			{\Lambda}^\prime(\tau) &= 1 - \left(1-\frac{1}{\delta}\right) \cdot \frac{\E G^2}{(\E G)^2}, \\ {\Lambda}^{\prime\prime}(\tau) &= 2\left(1-\frac{1}{\delta}\right) \cdot \frac{ (\E G^3)\cdot (\E G) - (\E G^2)^2 }{(\E G)^3}.
			\end{align*}
			Consider the following two cases,
			\paragraph{Case 1: $\tau \in [1,\infty)$.} Applying Chebychev's Association Inequality (Fact \ref{chebychev_association}) with $A=B=G$ and $f(a) = g(a) = a$ gives us that $\Lambda^{\prime \prime}(\tau) \geq 0$. In fact, an inspection of the proof of the Chebychev's Association Inequality from \cite{boucheron2013concentration} allows us to rule out the equality case under the assumptions imposed on $\mathcal{T}$, and we have $\Lambda^{\prime\prime}(\tau) > 0$. Hence, $\Lambda$ is strictly convex in $(1,\infty)$. Since $\Lambda(\tau)$ is continuous on $[1,\infty)$, we have $\Lambda$ is convex on $[1,\infty)$
			\paragraph{Case 2: $\tau \in (-\infty,0]$. } Again, applying Chebychev's Association Inequality with $A = B = -G$ and $f(a) = f(b) = a$ gives us ${\Lambda}^{\prime\prime}(\tau) \leq 0$, Hence $\Lambda$ is concave in this region. As before, an inspection of the proof of Chebychev's Association inequality allows us to rule out the equality case under the assumptions imposed on $\mathcal{T}$, and we have $\Lambda^{\prime\prime}(\tau) < 0$. Hence, $\Lambda$ is strictly concave in $(-\infty,0)$. Since $\Lambda(\tau)$ is continuous on $(-\infty,0)$, we have $\Lambda$ is concave on $(-\infty,0]$. This concludes the proof of statement (1) in the lemma.
			\item Note that, 
			\begin{align*}
			\lim_{\tau \rightarrow \infty} \tau - \frac{(1-1/\delta)}{\E\left[\frac{1}{\tau- T}\right]} & = \tau \left(1 - \frac{(1-1/\delta)}{\E\left[\frac{\tau}{\tau- T}\right]} \right) = \infty.
			\end{align*}
			This shows $\lim_{\tau \rightarrow \infty} \Lambda(\tau) = \infty$. The claim about the limit as $\tau \rightarrow - \infty$ can be analogously obtained. This proves item (2) in the statement of the lemma.
			\item The infimum in the definition of $\lambda_{r}$ is attained due to item (2) in the statement of the lemma. Analogously,  the supremum in the definition of $\lambda_{l}$ is attained. Next consider any $\tau_+ \in (1,\infty)$ and any $\tau_{-} \in (-\infty,0)$. Since the function $f(t) = (\tau_{+}-t)^{-1}$ is convex on $[0,1]$, according to Jensen's Inequality, we have
			\begin{align*}
			\Lambda(\tau_{+}) &\geq \tau_{+} - \left( 1 - \frac{1}{\delta} \right) \cdot (\tau_{+} - \E[T]) \\
			& = \frac{\tau_+}{\delta} +  \left( 1 - \frac{1}{\delta} \right) \cdot \E[T].
			\end{align*}
			On the other hand, since the function $f(t) = (\tau_{-}-t)^{-1}$ is concave on $[0,1]$, we have
			\begin{align*}
			\Lambda(\tau_{-}) &\leq \tau_{-} - \left( 1 - \frac{1}{\delta} \right) \cdot (\tau_{-} - \E[T]) \\
			& = \frac{\tau_-}{\delta} +  \left( 1 - \frac{1}{\delta} \right) \cdot \E[T].
			\end{align*}
Hence,
			\begin{align*}
			\Lambda({\tau_+}) &\geq \frac{1}{\delta} + \left( 1 - \frac{1}{\delta} \right) \cdot \E[T] \\&>  \left( 1 - \frac{1}{\delta} \right) \cdot \E[T] \\&\geq \Lambda(\tau_{-}).
			\end{align*}
			Taking the minimum over $\tau_{+}$ and maximum of $\tau_{-}$ gives us $\lambda_{r} > \lambda_{l}$. Furthermore we note that $\Lambda(0^{-}) \geq 0$. Hence $\lambda_{l} \geq 0$.  This concludes the proof of item (3) in the statement of the lemma.
			\item For any $x \in (\lambda_{l},\lambda_{r})$, $\Lambda(\tau) = x$ doesn't have a solution in $(-\infty,0] \cup [1,\infty)$ since $\Lambda(\tau) \leq \lambda_{l} \; \forall \; \tau \leq 0$ and $\Lambda(\tau) \geq \lambda_{r} \; \forall \; \tau \geq 1$. Now consider any $x \geq \lambda_{r}$. Since $\lambda(\tau) \leq \lambda_{l} < \lambda_{r} \; \forall \; \tau \leq 0$, we know that all solutions of $\Lambda(\tau) = x$ lie in $[1,\infty)$. Since $\Lambda$ is strictly convex in $(1,\infty)$, there can be atmost 2 solutions. Now consider any $x> \lambda_{r}$. Let $\tau_{r} = \arg\min_{\tau \geq 1} \Lambda (\tau)$.  Due to strict convexity of $\Lambda(\tau)$, we have $\Lambda^\prime(\tau) > 0$ for any $\tau \in (\tau_{r},\infty)$. Hence $\Lambda(\tau)$ is strictly increasing on $[\tau_{r}, \infty)$. Since $\lambda_{r} = \Lambda(\tau_{r}) < x < \Lambda(\infty) = \infty $, we are guaranteed to have exactly one solution to $\Lambda(\tau) = x$ on $(\tau_{r},\infty)$ which indeed satisfies $\Lambda^\prime(\tau) > 0$. The analysis for the case when $x \leq \lambda_{l}$ can be done in a similar way. This concludes the proof of item (4) in the statement of the lemma.
		\end{enumerate}
	\end{proof}
	We are now in the position to characterize the support of $\gamma \boxtimes T$ which is the content of the following proposition.
	\begin{proposition} \label{support_prop}The support of $\gamma \boxtimes T$ is given by
		\begin{align*}
		\text{Supp}(\gamma \boxtimes T) & = [\lambda_{l},\lambda_{r}] \cup \text{Supp}((\gamma \boxtimes T)_d),
		\end{align*}
		where $(\gamma \boxtimes T)_d$ denotes the discrete part of the measure $\gamma \boxtimes T$. If the random variable $T$ has a density with respect to the Lebesgue measure, then, 
		\begin{align*}
		    \text{Supp}(\gamma \boxtimes T) & = [\lambda_{l},\lambda_{r}].
		\end{align*}
	\end{proposition}
\begin{proof}
We first claim that $(\lambda_{l}, \lambda_{r}) \subset \text{Supp}(\gamma \boxtimes T)$. Since the support of a measure is closed, this means that $[\lambda_{l}, \lambda_{r}] \subset \text{Supp}(\gamma \boxtimes T)$. We prove this claim by contradiction. Suppose that $\exists \lambda \in (\lambda_{l},\lambda_{r})$ such that $\lambda \not\in \text{Supp}(\gamma \boxtimes T)$. To simplify notation, for $z \in \mathbb C^{-}$, we introduce the following reciprocal subordination function $\tau_T(z)$
		\begin{align*}
		\tau_T(z) & \Mydef \frac{1}{w_T(1/z)}.
		\end{align*}
		According to Lemma \ref{extension_lemma}, we have
		\begin{align*}
		\tau_T(\lambda) \Mydef \lim_{\epsilon \rightarrow 0^{+}} \tau_T(\lambda - i \epsilon)  \in (-\infty,0) \cup (1,\infty).
		\end{align*}
		By Lemma \ref{formula_for_cauchy_and_subordination}, $\tau_T(\lambda - i \epsilon)$ uniquely solves the equation $\Lambda(\tau) = \lambda - i \epsilon$ in $\mathbb C^{-}$. Taking $\epsilon \rightarrow 0$, we obtain,
		\begin{align*}
		\lambda & = \lim_{\epsilon \rightarrow 0^+} \Lambda(\tau_T(\lambda - i \epsilon)) \\
		& = \lim_{\epsilon \rightarrow 0^{+}} \left( \tau_T(\lambda - i \epsilon) - \frac{1-1/\delta}{\E \left[ \frac{1}{\tau_T(\lambda - i \epsilon)- T} \right]} \right) \\
		& \explain{(a)}{=} \tau_T(\lambda) - \frac{1-1/\delta}{\E \left[ \frac{1}{\tau_T(\lambda)- T} \right]}.
		\end{align*}
		In the step marked (a), we used the fact that since $\lim_{\epsilon \rightarrow 0^{+}} \tau_T(\lambda - i \epsilon) \not\in \text{Supp}(T)$, we have $\exists c>0$, such that for any $\epsilon$ small enough $\text{dist}(\tau_T(\lambda- i \epsilon), \text{Supp}(T)) \geq c$. This gives us a dominating function for an application of  the dominated convergence theorem. Hence, we have found a solution for the equation $\lambda = \Lambda(\tau), \tau \in (-\infty,0) \cup (1,\infty)$. But this contradicts Lemma \ref{solutions_real_case}. Hence, we have, $(\lambda_{l}, \lambda_{r}) \subset \text{Supp}(\gamma \boxtimes T)$.
		
		Next, we claim that any $x \in [0,\lambda_{l}) \cup (\lambda_{r},\infty)$ is not in the support of the absolutely continuous part of $\gamma \boxtimes T$. To show this, we first compute a first order asymptotic expansion of $\tau_T(x - i \epsilon)$ for $\epsilon \approx 0$. From Lemma \ref{solutions_real_case}, we know there exists a unique solution for the equation $\Lambda(\tau) = x, \tau \in (-\infty,0) \cup (1,\infty)$ and ${\Lambda}^\prime(\tau) > 0$. We denote this solution by $\tau_\star$. Since $\tau_\star \not\in \text{Supp}(T)$, the function $\Lambda(\tau)$ is analytic in the neighborhood (in $\mathbb C$) of $\tau_\star$. The implicit function theorem guarantees us a solution $\tau(\epsilon) = \tau_R(\epsilon) + i \tau_I(\epsilon)$ of the equation $\Lambda(\tau) = x - i\epsilon$. However, this $\tau(\epsilon)$ may not be the reciprocal subordination function $\tau_T(x-i\epsilon)$ since we still need to verify it is in $\mathbb C^-$. To take care of this, again by the implicit function theorem we have
		\begin{align*}
		\Lambda^\prime(\tau_\star) \cdot \frac{\diff \tau}{\diff \epsilon}(0) & = -i.
		\end{align*}
		This gives us
		\begin{align*}
		\frac{\diff \tau_I}{\diff \epsilon}(0) = - \frac{1}{{\Lambda}^\prime(\tau_\star)} < 0, \; \frac{\diff \tau_R}{\diff \epsilon}(0) = 0 . 
		\end{align*}
		Hence, we have 
		\begin{align*}
		\tau(\epsilon) & = \tau_\star - i \frac{\epsilon}{{\Lambda}^\prime(\tau_\star)} + o(\epsilon).
		\end{align*}
		This verifies that $\tau(\epsilon) \in \mathbb C^{-}$ for $\epsilon$ small enough. Finally since $\tau_T(x-i\epsilon)$ is the unique solution to the equation $\Lambda(\tau) = x - i \epsilon$ in $\mathbb C^{-}$, we have 
		\begin{align*}
		\tau_T(x-i\epsilon) & = \tau_\star - i \frac{\epsilon}{{\Lambda^\prime}(\tau_\star)} + o(\epsilon).
		\end{align*}
		According to the Stieltjes Inversion Formula, Theorem \ref{stieltjes_theorem}, we obtain
		\begin{align*}
		\frac{\diff (\gamma \boxtimes T)_{ac}}{\diff x}(x) & = \frac{1}{\pi x} \cdot \Im \Bigg( \frac{1-\frac{1}{\delta}}{1-x \cdot \lim_{\epsilon \rightarrow 0^{+}} w_T\left(\frac{1}{x-i \epsilon}\right)} \Bigg) \\
		& \explain{(b)}{=} \frac{1}{\pi x} \cdot \Im \left( \frac{(1-1/\delta) \cdot \tau_\star}{\tau_\star - x} \right)  = 0.
		\end{align*}
		In the step marked (b), we are relying on the assumption that $\tau_\star \neq x$. To verify this, we recall that $\tau_\star$ solves, $\Lambda(\tau_\star) = x$ and $\tau_\star \not \in [0,1]$. This means that
		\begin{align*}
		|\tau_\star - x| &= \frac{1-1/\delta}{\left|\E \left[ \frac{1}{\tau_\star-T} \right]\right|} \\ &\geq \frac{1-1/\delta}{\E \left[ \left| \frac{1}{\tau_\star-T} \right|\right]} \\ &\geq (1-1/\delta) \cdot \text{dist}(\tau_\star,[0,1])>0.
		\end{align*}
		Hence, we have shown
		\begin{align*}
		\frac{\diff (\gamma \boxtimes T)_{ac}}{\diff x}(x)  \explain{a.s.}{=} 0, \forall x \in [0,\lambda_{l}) \cup (\lambda_{r} ,\infty).
		\end{align*}
		This implies,
		\begin{align*}
		     [0,\lambda_{l}) \cup (\lambda_{r} ,\infty)  \subset \mathbb R \backslash \text{Supp}((\gamma \boxtimes T)_{ac}).
		\end{align*}
		Taking complements, we have $\text{Supp}((\gamma \boxtimes T)_{ac}) \subset [\lambda_{l},\lambda_{r}]$.
		Hence, we have shown that
		\begin{align*}
		&[\lambda_{l},\lambda_{r}] \cup \text{Supp}((\gamma \boxtimes T)_d)   \subset \text{Supp}(\gamma \boxtimes T) \\ & \hspace{3cm}= \text{Supp}((\gamma \boxtimes T)_{ac}) \cup \text{Supp}((\gamma \boxtimes T)_d) \\& \hspace{3cm}\subset  [\lambda_{l},\lambda_{r}] \cup \text{Supp}((\gamma \boxtimes T)_d).
		\end{align*}
		Therefore, $\text{Supp}(\gamma \boxtimes T)=[\lambda_{l},\lambda_{r}] \cup \text{Supp}((\gamma \boxtimes T)_d)$ which proves the claim of the proposition. Finally, when $T$ has a density with respect to Lebesgue measure, Theorem \ref{regularity_theorem} gives us $\text{Supp}((\gamma \boxtimes T)_d) = \emptyset$ which yields the second claim in the proposition. 
\end{proof}

Finally we note that in order to apply Theorem \ref{outlier_theorem_free_probability}, it is necessary to understand the set ${\tau_T^{-1}(\{\theta\}) \cap (\mathbb R \backslash {\text{Supp}(\gamma \boxtimes T)}})$,  $\theta \in \mathbb R$ (See Theorem \ref{outlier_theorem_free_probability} for a definition of $\tau_T$). This is done in the following lemma.

\begin{lemma} Let $(w_\gamma, w_T)$ denote the subordination functions corresponding to the free multiplicative convolution of $\gamma, \mathcal{L}_T$. Define
	\begin{align*}
	\tau_T(z) & = \frac{1}{w_T(1/z)}.
	\end{align*}
	Then, we have
	\begin{align*}
	{\tau_T^{-1}(\{\theta\}) \cap (\mathbb R \backslash {\text{Supp}(\gamma \boxtimes T)}}) & = \begin{cases} \theta \in [\tau_{l},\tau_{r}] : & \emptyset \\
	\theta \not\in [\tau_{l},\tau_{r}]: & \{\Lambda(\theta)\}
	\end{cases},
	\end{align*}
\label{tau_inv_func}
where where,  $\tau_{l} \triangleq \arg\max_{\tau \leq 0} \Lambda(\tau)$, $\tau_{r} \triangleq \arg\min_{\tau \geq 1} \Lambda(\tau)$.
\end{lemma}
\begin{proof} From Proposition \ref{support_prop}, we know that $\text{Supp}(\gamma \boxtimes T) = [\lambda_{l},\lambda_{r}]$, where $\lambda_{l}  \Mydef \max_{\tau \leq 0} \Lambda(\tau)$ and $\lambda_{r}  \Mydef \min_{\tau \geq 1} \Lambda(\tau)$. Furthermore, we showed that for any $x \not \in [\lambda_{l},\lambda_{r}]$, the reciprocal subordination function $\tau_T(x)$ is the unique solution to the equations: $\Lambda(\tau) = x, \Lambda^\prime(\tau) > 0, \; \tau \not\in [0,1]$.
	From Lemma \ref{solutions_real_case}, we know that when $x> \lambda_{r}$, the unique solution to $\Lambda(\tau) = x, \Lambda^\prime(x) > 0$ satisfies $\tau > \tau_{r}$ and when $x<\lambda_{l}$, the unique solution satisfies $\tau < \tau_{l}$. These considerations immediately yield the claim of the lemma.
\end{proof}
\subsection{Proof of Lemmas \ref{L_m_func_asymptotic} and \ref{Lem:Lambda_plus}}\label{Sec:technical_lemmas}

Recall we defined $\Lambda_+(\tau)$ as
\begin{align*}
\Lambda_+(\tau) & = \begin{cases} \tau - \frac{(1-1/\delta)}{\E\left[\frac{1}{\tau- T}\right]} \ \  \ \ \ \ \ \ \ \ \ \ \ \ \ \   {\rm if} \ \tau > \tau_{r},  \\
\min_{\tau \geq 1} \left(  \tau - \frac{(1-1/\delta)}{\E\left[\frac{1}{\tau- T}\right]} \right) \ \ {\rm if} \ \tau \leq \tau_{r},
\end{cases}
\end{align*}
where $T=\T(|Z|/\sqrt{\delta})$ and $Z\sim\mathcal{CN}(0,1)$, and
\begin{align*}
\tau_{r} & \triangleq \arg \min_{\tau \geq 1} \left(  \tau - \frac{(1-1/\delta)}{\E\left[\frac{1}{\tau- T}\right]} \right).
\end{align*}
We first prove Lemma \ref{L_m_func_asymptotic}, which we restated below for convenience. \vspace{5pt}

\textbf{Lemma 3.} \textit{Let $\vartheta_c  \Mydef \left(1 - \left( \E \left[ \frac{|Z|^2}{1-T} \right] \right)^{-1} - \E[|Z|^2 T]  \right)^{-1}$. Define the function $\theta(\vartheta)$ as:
	\begin{itemize}
		\item {When $\vartheta> \vartheta_c$: } Let $\theta(\vartheta)$ be the unique value of $\lambda$ that satisfies the equation: $$\lambda - \E[|Z|^2 T] - 1/\vartheta  = \left(\E \left[ \frac{|Z|^2}{\lambda - T} \right]\right)^{-1},$$ in the interval:
		\begin{align*}
		    \lambda \in  \left( \max(1, \E[ |Z|^2 T]+1/\vartheta), \infty \right).
		\end{align*}
		\item {When $\vartheta \leq \vartheta_c$: } $\theta(\vartheta) \Mydef 1$.
	\end{itemize}
Then, we have $L_m(\vartheta) \explain{a.s.}{\rightarrow} \Lambda_{+} (\theta(\vartheta))$, where $L_m(\vartheta) $ is defined in \eqref{L_m_eq_numbered}.}
\begin{proof}
In Proposition \ref{spectrum_E_empirical}, we obtained an asymptotic characterization of the spectrum of $\bm E(\vartheta)$. More specifically, we proved that
	\begin{align*}
	\mu_{\bm E(\vartheta)} \explain{d}{\rightarrow} \mathcal{L}_T, \; \lambda_1(\bm E(\vartheta)) \rightarrow \theta(\vartheta).
	\end{align*}
	We recall the matrix $\bm R$ was defined as
	\begin{align*}
	\bm R &= \begin{bmatrix} \bm I_{n-1} & \bm 0_{n-1,m-1} \\
	\bm 0_{m-n,n-1} & \bm 0_{m-1,m-1}
	\end{bmatrix}.
	\end{align*}
	In particular, $\mu_{\bm R} \explain{d}{\rightarrow} \gamma$, where the measure $\gamma$ is given by
	\begin{align*}
	\gamma & = \frac{1}{\delta} \delta_1 + \left( 1 - \frac{1}{\delta} \right) \delta_0.
	\end{align*}
	Applying Theorem \ref{outlier_theorem_free_probability}, we obtain: 
	\begin{enumerate}
		\item The spectral measure of $\bm E(\vartheta) \bm H_{m-1} \bm R \bm H_{m-1}^\UH$ converges to:
		\begin{align*}
		\mu_{\bm E(\vartheta) \bm H_{m-1} \bm R \bm H_{m-1}^\UH} \explain{d}{\rightarrow} \gamma \boxtimes \mathcal{L}_T.
		\end{align*}
		\item For any $\epsilon > 0$, we have, almost surely, for $m$ large enough that, $\sigma(\bm E(\vartheta) \bm H_{m-1} \bm R \bm H_{m-1}^\UH) \subset K_\epsilon$, where $K_\epsilon$ is the $\epsilon$-neighborhood of the set $K = \text{Supp}(\gamma \boxtimes \mathcal{L}_T) \cup \tau_T^{-1}(\{\theta(\vartheta)\})$.
		\item For any $\lambda \in \tau_T^{-1}(\{\theta(\vartheta)\}) \cap (\mathbb R \backslash \text{Supp}(\gamma \boxtimes \mathcal{L}_T)) $, we have almost surely exactly one eigenvalue of $\bm E(\vartheta) \bm H_{m-1} \bm R \bm H_{m-1}^\UH$ in a small enough neighborhood of $\lambda$ for large enough $n$.
	\end{enumerate}
	In Proposition \ref{support_prop}, we characterized $\text{Supp}(\gamma \boxtimes \mathcal{L}_T)$ as $[\lambda_{l},\lambda_{r}]$, where $\lambda_{l} = \max_{\tau \leq 0} \Lambda(\tau)$, $\lambda_{r} = \min_{\tau \geq 1} \Lambda(\tau)$ and the function $\Lambda(\tau)$ is given by: $$\Lambda(\tau) =  \tau - \frac{(1-1/\delta)}{\E\left[\frac{1}{\tau- T}\right]}. $$
	In Lemma \ref{tau_inv_func}, we characterized the set:
	\begin{align*}
	\tau_T^{-1}(\{\theta\}) \cap (\mathbb R \backslash {\text{Supp}(\gamma \boxtimes T)}) & = \begin{cases} \emptyset \  &\theta \in [\tau_{l},\tau_{r}], \\
	 \{\Lambda(\theta) \} \ &\theta \not\in [\tau_{l},\tau_{r}],
	\end{cases}
	\end{align*}
	where,  $\tau_{l} \triangleq \arg\max_{\tau \leq 0} \Lambda(\tau)$, $\tau_{r} \triangleq \arg\min_{\tau \geq 1} \Lambda(\tau)$. Putting these together, one obtains the following two cases:
	\begin{description}
		\item {Case 1: $\theta(\vartheta) \leq \tau_{r}.$} In this case, the set $\tau_T^{-1}(\{\theta\}) \cap (\mathbb R \backslash {\text{Supp}(\gamma \boxtimes T)}) = \emptyset$. The matrix $\bm E(\vartheta) \bm H_{m-1} \bm R \bm H_{m-1}^\UH$ has no eigenvalues outside the support of the bulk distribution, and 
		\begin{align*}
		L_m(\vartheta) \explain{a.s.}{\rightarrow} \lambda_{r} = \Lambda(\tau_{r}).
		\end{align*}
		\item {Case 2: $\theta(\vartheta) > \tau_{r}.$} In this case, the set $$\tau_T^{-1}(\{\theta\}) \cap (\mathbb R \backslash {\text{Supp}(\gamma \boxtimes T)}) = \{\Lambda(\theta(\vartheta))\}.$$ Hence, there is an eigenvalue in the neighborhood of $\Lambda(\theta(\vartheta)))$. Since $\theta(\vartheta) > \tau_{r}$, and $\Lambda$ is a strictly increasing function on $[\tau_{r},\infty)$ (Lemma \ref{solutions_real_case}), we have $\Lambda(\theta(\vartheta)) > \lambda_{r}$. Hence the eigenvalue in the neighborhood of $\Lambda(\theta(\vartheta))$ is the largest one, and we have 
		\begin{align*}
		L_m(\vartheta) \explain{a.s.}{\rightarrow} \Lambda(\theta(\vartheta)).
		\end{align*}
	\end{description}
	It is now straightforward to check that the above two cases can be combined into a concise form stated in the claim of the lemma. 
\end{proof}
We end this section by proving Lemma \ref{Lem:Lambda_plus}, restated below for convenience.\vspace{5pt}

\textbf{Lemma 4.} \textit{The following hold for the equation:  $$\Lambda_{+}(\theta(\vartheta)) = 1/\vartheta + \E[|Z|^2 T], \; \vartheta>0.$$
	\begin{enumerate}
		\item This equation has a unique solution. 
		\item Let $\vartheta_\star$ denote the solution of the above equation. Then: 
	\end{enumerate}
			\paragraph{Case 1 } If $ \psi_1(\tau_{r}) \leq \frac{\delta}{\delta - 1},$ we have $$\Lambda_+(\theta(\vartheta_\star)) = \Lambda(\tau_r).$$ Furthermore if $\psi_1(\tau_{r}) < \delta/(\delta - 1)$, then,
			\begin{align*}
			\frac{\diff \Lambda_+(\theta(\vartheta))}{\diff \vartheta} \bigg\rvert_{\vartheta = \vartheta_\star} = 0,
			\end{align*}
			\paragraph{Case 2 } If $ \psi_1(\tau_{r}) > \frac{\delta}{\delta - 1},$ we have $$\Lambda_+(\theta(\vartheta_\star)) = \Lambda(\theta_\star),$$ and, 
			\begin{multline*}
			\frac{\diff \Lambda_+(\theta(\vartheta))}{\diff \vartheta} \bigg\rvert_{\vartheta = \vartheta_\star} = \\  \frac{1}{\vartheta_\star^2} \cdot \frac{\delta}{\delta-1} \cdot \left( \frac{\delta}{\delta -1 } - \psi_2(\theta_\star) \right) \cdot \frac{1}{\psi_3^2(\theta_\star) -\frac{\delta^2}{(\delta-1)^2}}.
			\end{multline*} where $\theta_\star > 1$ is the unique  $\theta \geq \tau_{r}$ that satisfies
			$\psi_1(\theta) = \frac{\delta}{\delta - 1}.$}
\begin{proof}
	Before we begin the proof of this lemma, it is helpful to list the conclusions of some of the previous lemmas.
	\begin{description}
		\item {\underline{Lemma \ref{Q_eq_population}}:} In this lemma, for $\vartheta > \vartheta_c$  we defined the function $\theta(\vartheta)$ as the unique value of $\lambda >  \max(1, \E[ |Z|^2 T]+1/\vartheta)$ that satisfies 
		\begin{align*}
		\lambda - \E[|Z|^2 T] - 1/\vartheta & = \frac{1}{\E \left[ \frac{|Z|^2}{\lambda - T} \right]}.
		\end{align*}
		We also set  $\theta(\vartheta) = 1$ when $\vartheta \leq \vartheta_c$. We also showed that $\theta(\vartheta)$ is strictly increasing on $[\vartheta_c, \infty)$ and $\theta(\infty) = \infty$. In particular $\theta(\vartheta)$ has a well defined inverse defined on the domain $[1,\infty)$ given by: 
		\begin{align}
		\theta^{-1}(\lambda) & =  \left( \lambda - \E[|Z|^2 T] -\frac{1}{\E \left[ \frac{|Z|^2}{\lambda-T}\right]} \right)^{-1}.
		\label{theta_inv_func_eq}
		\end{align}
		\item {\underline{Lemma \ref{solutions_real_case}}:} We defined the function $\Lambda(\tau)$ as
		\begin{align}
		\Lambda(\tau) & \triangleq  \tau - \frac{(1-1/\delta)}{\E \left[ \frac{1}{\tau - T} \right]}.
		\label{Lambda_def_eq}
		\end{align}
		We showed the that $\Lambda(\tau)$ is strictly convex on $[1,\infty)$. We defined $(\tau_{r}, \lambda_{r})$ to be the minimizing argument and the minimum value of $\Lambda(\tau)$ in $[1,\infty)$. In particular $\tau_{r} \geq 1$. We also showed that $\Lambda(\infty) = \infty$.
		We further defined $\Lambda_+(\tau)$ in the following way: 
		\begin{align*}
		\Lambda_{+}(\tau) & = \begin{cases} \lambda_{r},  & \tau \leq \tau_{r}. \\
		\Lambda(\tau), & \tau > \tau_{r}.
		\end{cases}
		\end{align*}
	\end{description}
	Some simple implications of the above assertions are: First, since $\theta(\vartheta)$ and $\Lambda_+$ are both non-decreasing continuous functions $\Lambda_+(\theta(\vartheta))$ is non-decreasing and continuous. Second, since $\Lambda(\tau) = \lambda_{r}$ for $\tau \leq \tau_{r}$, we have, for all $\vartheta \leq \theta^{-1}(\tau_{r})$, $\Lambda_+(\theta(\vartheta)) = \lambda_{r}$. Third since  $\theta(\infty) = \infty$ and $\Lambda(\infty) = \infty$, we have, $\Lambda_+(\theta(\vartheta)) \rightarrow \infty$ as $\vartheta \rightarrow \infty$. The only possible point of non-differentiability of $\Lambda_+(\theta(\vartheta))$ is at $\vartheta = \theta^{-1}(\tau_{r})$.  It is straightforward to compute the derivative of $\Lambda(\theta(\vartheta))$ at all other points using implicit function theorem and obtain 
	\begin{align}
	\frac{\diff \Lambda_+(\theta(\vartheta))}{\diff \vartheta} & = \begin{cases} 0 & \vartheta < \theta^{-1}(\tau_{r}), \\ 
	 \Lambda^\prime(\theta(\vartheta)) \cdot \theta^\prime(\vartheta) & \vartheta > \theta^{-1}(\tau_{r}). \end{cases}
	\label{deriv_eq}
	\end{align}
	The derivatives of $\Lambda, \theta$ can be calculated as, 
	\begin{align}
	\Lambda^\prime(\tau) & = \frac{\delta - 1}{\delta} \left( \frac{\delta}{\delta - 1} -\psi_2(\tau) \right). \label{deriv_eq_supp}\\
	\theta^\prime(\vartheta) & = \frac{1}{\vartheta^2} \left( \frac{ \left(\E\left[ \frac{|Z|^2}{\theta(\vartheta) - T} \right] \right)^2}{ \E\left[ \frac{|Z|^2}{(\theta(\vartheta) - T)^2} \right] - \left(\E\left[ \frac{|Z|^2}{\theta(\vartheta) - T} \right] \right)^2} \right). \label{deriv_eq_supp_2}
	\end{align}
	A representative plot of the function $\Lambda_+(\theta(\vartheta))$ is shown in Figure \ref{Lambda_plus_eq}. 
	
\begin{figure*}[hbpt]
\centering
\includegraphics[width=\textwidth]{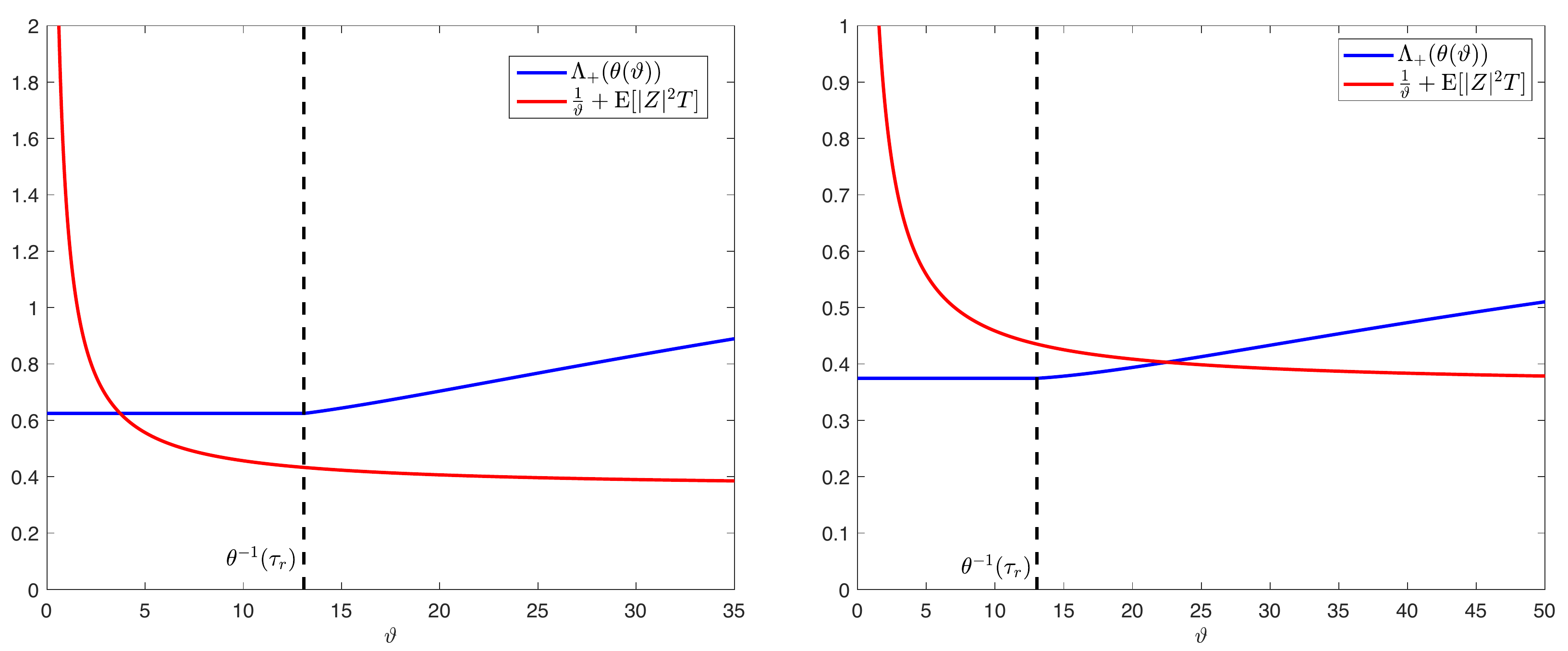}
			\caption{Typical Plots of the functions $\Lambda_+(\theta(\vartheta))$ (Blue) and $\E[|Z|^2 T] + \frac{1}{\vartheta}$ (Red). Case 1 (Left): The two functions intersect at the constant part of $\Lambda_+(\theta(\vartheta))$, Case 2 (Right): The The two functions intersect at the increasing part of $\Lambda_+(\theta(\vartheta))$} 
\label{Lambda_plus_eq}
\end{figure*}

%
	We are now in a position to prove the claims of the lemma. 
	
	\begin{enumerate}
		\item Since $\Lambda_{+}(\theta(\vartheta))$ is continuous and non-decreasing and $1/\vartheta + \E[|Z|^2T] $ is continuous and strictly decreasing, the fixed point equation can have at most one solution. On the other hand comparing the values of the two sides of the fixed point equation at $\vartheta \rightarrow 0$ and $\vartheta \rightarrow \infty$ shows that there is at least one solution. 
		\item  Let $\vartheta_\star$ be denote the solution of the fixed point equation $\Lambda_+(\theta(\vartheta)) = 1/\vartheta + \E[|Z|^2T] $. A typical plot of these two functions is shown in Figure \ref{Lambda_plus_eq}. The figure shows two possible cases for the intersection of the two curves: 
			\textit{Case 1: } The curves intersect at a point $\vartheta_\star \leq \theta^{-1}(\tau_{r})$ (or on the flat part of $\Lambda_+(\theta(\alpha)$). In this case we have, $\Lambda_+(\theta(\vartheta_\star)) = \lambda_{r}$.
			
			\textit{Case 2: } The curves intersect at a point $\vartheta_\star > \theta^{-1}(\tau_{r})$ or the rising part of $\Lambda_+(\theta(\alpha)$. We have $\Lambda_+(\theta(\vartheta_\star)) > \lambda_{r}$. 
		We can distinguish between the two cases by comparing the value of the function $1/\vartheta + \E[|Z|^2 T]$ at $\vartheta = \theta^{-1}(\tau_{r})$ with $\lambda_{r}$. In particular, we have, 
		
			\textit{Case 1: } $$\Lambda_+(\theta(\vartheta_\star)) = \lambda_{r} \Leftrightarrow 1/\theta^{-1}(\tau_{r}) + \E[|Z|^2 T] \leq \lambda_{r},$$
			
			\textit{Case 2: } $$\Lambda_+(\theta(\vartheta_\star)) > \lambda_{r} \Leftrightarrow 1/\theta^{-1}(\tau_{r}) + \E[|Z|^2 T]> \lambda_{r}.$$ 

		Substituting the formula for $\theta^{-1}(\tau_{r})$, mentioned in \eqref{theta_inv_func_eq}, and $\lambda_{r} = \Lambda(\tau_{r})$ and the formula for $\Lambda$ from \eqref{Lambda_def_eq}, the 2 cases can be simplified slightly more.

			\textit{Case 1: } This case occurs when
			$$\frac{1}{\theta^{-1}(\tau_{r})} + \E[|Z|^2 T] \leq \lambda_{r} \Leftrightarrow \frac{\E \left[ \frac{|Z|^2}{\tau_{r} - T} \right]}{\E \left[ \frac{1}{\tau_{r} - T} \right]} \leq \frac{\delta}{\delta - 1}.$$
			In this situation, we have, $\Lambda_+(\theta(\vartheta_\star)) = \lambda_{r}$. Furthermore, if we additionally have
			\begin{align*}
			\frac{\E \left[ \frac{|Z|^2}{\tau_{r} - T} \right]}{\E \left[ \frac{1}{\tau_{r} - T} \right]} < \frac{\delta}{\delta - 1}
			\end{align*}
			Then $\Lambda_+(\theta(\vartheta))$ is differentiable at $\vartheta_\star$ and, from \eqref{deriv_eq}, we have
			\begin{align*}
			\frac{\diff \Lambda_+(\theta(\vartheta))}{\diff \vartheta} \bigg\rvert_{\vartheta = \vartheta_\star} & = 0.
			\end{align*}
			
			\textit{Case 2: } This case occurs when
			\begin{align*}&\frac{1}{\theta^{-1}(\tau_{r})} + \E[|Z|^2 T] > \lambda_{r} \\ & \hspace{2cm} \Leftrightarrow \frac{\E \left[ \frac{|Z|^2}{\tau_{r} - T} \right]}{\E \left[ \frac{1}{\tau_{r} - T} \right]} > \frac{\delta}{\delta - 1}.
			\end{align*}
			In this situation, we have, $\Lambda_+(\theta(\vartheta_\star)) > \lambda_{r}$. It turns out that we can give a simpler expression for $\Lambda_+(\theta(\vartheta_\star))$. In this case,
			$\vartheta_\star \geq \theta^{-1}(\tau_{r})$ solves,
			\begin{align}
			\Lambda(\theta(\vartheta_\star)) & = \frac{1}{\vartheta_\star} + \E[|Z|^2 T],  
			\label{elim_eq_1}
			\end{align}
			and $\theta(\vartheta_\star) \geq 1$  is the solution of the equation
			\begin{align}
			&\E[|Z|^2 T] + \frac{1}{\vartheta_\star}  = \theta(\vartheta_\star) - \frac{1}{\E \left[ \frac{|Z|^2}{\theta(\vartheta_\star) - T} \right]}.
			\label{elim_eq_2}
			\end{align}
			By definition the function $\Lambda(\tau(\alpha))$ is 
			\begin{align}
			\Lambda(\theta(\vartheta_\star)) & =  \theta(\vartheta_\star) - \frac{(1-1/\delta)}{\E \left[ \frac{1}{\theta(\vartheta_\star) - T} \right]}.
			\label{elim_eq_3}
			\end{align}
			We first eliminate $\vartheta_\star$ from Equations \eqref{elim_eq_1}-\eqref{elim_eq_3} and conclude that $\theta_\star \Mydef \theta(\vartheta_\star)$ solves
			\begin{align}
			\frac{\E \left[ \frac{|Z|^2}{\theta_\star - T} \right]}{\E \left[ \frac{1}{\theta_\star - T} \right]} = \frac{\delta}{\delta - 1}, \; \theta_\star \geq \tau_{r},
			\label{theta_star_eq}
			\end{align}
			and $\vartheta_\star$ is given by
			\begin{align*}
			\vartheta_\star & = \left( \theta_\star -  \frac{1}{\E \left[ \frac{|Z|^2}{\theta_\star - T} \right]}- \E[|Z|^2 T] \right)^{-1}.
			\end{align*}
			Since the solution to Equations \eqref{elim_eq_1}-\eqref{elim_eq_3} was guaranteed to be unique, the solution to \eqref{theta_star_eq} is guaranteed to be unique. Finally we can compute the derivative of $\Lambda_+(\theta(\vartheta))$ at $\vartheta = \vartheta_\star$. It will be convenient  to introduce the random variable $G = (\theta_\star - T)^{-1}$ to write the equations in a compact form.  From \eqref{deriv_eq}-\eqref{deriv_eq_supp_2}, we have
			\begin{align*}
			&\frac{\diff \Lambda_+(\theta(\vartheta))}{\diff \vartheta} \bigg\rvert_{\vartheta = \vartheta_\star} = \Lambda^{\prime}(\theta_\star) \cdot \theta^\prime(\vartheta_\star) \\
			& \explain{}{=}  \frac{\delta - 1}{\delta \vartheta_\star^2} \left( \frac{\delta}{\delta -1 } - \psi_2(\theta_\star) \right)  \frac{\E[|Z|^2 G]^2}{\E[|Z|^2 G^2] - \E[|Z|^2 G]^2} \\
			& \explain{(a)}{=}  \frac{\delta \cdot \left( \frac{\delta}{\delta -1} - \psi_2(\theta_\star) \right)}{\vartheta_\star^2 \cdot (\delta-1) \cdot \psi_1^2(\theta_\star)}  \cdot   \frac{\E[|Z|^2 G]^2}{\E[|Z|^2 G^2] - \E[|Z|^2 G]^2} \\
			& \explain{}{=}\frac{\delta \cdot \left( \frac{\delta}{\delta -1} - \psi_2(\theta_\star) \right)}{\vartheta_\star^2\cdot(\delta-1)} \cdot  \frac{\E[G]^2}{\E[|Z|^2 G^2] - \E[|Z|^2 G]^2}\\
			& \explain{}{=} \frac{\delta}{\vartheta_\star^2(\delta-1)} \left( \frac{\delta}{\delta -1 } - \psi_2(\theta_\star) \right)  \frac{1}{\psi_3^2(\theta_\star) -\frac{\delta^2}{(\delta-1)^2}}.
			\end{align*}
			In the above display, in the step marked (a) we used the fact that $\theta_\star$ satisfies $\psi_1(\theta_\star) = \delta/(\delta-1)$.
			This concludes the proof of the characterization (2) given in the statement of the lemma.
	\end{enumerate}
\end{proof}

\section{Conclusions}\label{sec:conclusion}
We analyzed the asymptotic performance of a spectral method for phase retrieval under a random column orthogonal matrix model. Our results provides a rigorous justification for the conjectures in \cite{ma2019spectral}, which were obtained by analyzing an expectation propagation algorithm.  
%


%

\appendices


\section{Proof of Proposition \ref{optimal_prop}} \label{proof_optimal_trimming}
This section is devoted to the proof of Proposition \ref{optimal_prop}. We denote the functions $\Lambda,\psi_1,\psi_2,\psi_3$ (recall \eqref{key_functions}) with ${\T = \T_{\opt}}$ as $\Lambda_\opt, \psi_1^\opt, \psi_2^\opt, \psi_3^\opt$ and those with $\T = \T_{\opt,\epsilon}$ as $\Lambda_{\epsilon}, \psi_1^{\epsilon}, \psi_2^{\epsilon}, \psi_3^{\epsilon}$. Define the random variables:
\begin{align*}
    Z \sim \cgauss{0}{1}, \; T_\opt & = \T_{\opt}(|Z|/\sqrt{\delta}), \; T_{\epsilon} = \T_{\opt,\epsilon}(|Z|/\sqrt{\delta}).
\end{align*}

Next we observe that the function $\T_{\opt,\epsilon}$ is a bounded, strictly increasing, Lipchitz function and consequently $T_{\epsilon}$ has a density with respect to the Lebesgue measure. Hence by the rescale and shift argument outlined in Remark \ref{supp_remark}, Theorem \ref{result_lambda1} applies to a equivalent modification of $\T_{\opt,\epsilon}$ which can used to infer the corresponding result for $\T_{\opt,\epsilon}$ (after another rescale and shift argument). This gives us the result:

\begin{align} \label{reg_trimming_result}
    \frac{|\bm x_\star^\UH \hat{\bm x}_\epsilon|^2}{n} & \explain{a.s.}{\rightarrow}  \begin{cases} 0, & \psi_1^{\epsilon}(\tau_{r}^{\epsilon}) < \frac{\delta}{\delta - 1}, \\
	\frac{\left(\frac{\delta}{\delta -1 } \right)^2 - \frac{\delta}{\delta -1 } \cdot \psi_2^{\epsilon}(\theta_\star^\epsilon)}{\psi_3^\epsilon(\theta_\star^\epsilon)^2 - \frac{\delta}{\delta -1 } \cdot \psi_2^\epsilon(\theta_\star^\epsilon)}, & \psi_1(\tau_{r}^\epsilon) > \frac{\delta}{\delta - 1} . 
	\end{cases},
\end{align}
where $\tau_r^\epsilon \Mydef \argmin_{\tau \in [1,\infty)} \Lambda_\epsilon(\tau)$ and $\theta_\star^\epsilon$ is the solution to the fixed point equation (in $\tau$): $\psi_1^\epsilon(\tau) = \delta/(\delta-1)$ which is guaranteed to exist uniquely provided $\psi_1(\tau_{r}^\epsilon) > \delta/{(\delta - 1)}$. 
First we observe that,
\begin{align*}
    \Lambda^\prime_\epsilon(\tau)& = 1 - \left(1-\frac{1}{\delta}\right) \cdot \frac{\E G^2_\epsilon(\tau)}{(\E G_\epsilon(\tau))^2}, \; G_\epsilon(\tau) = (\tau-T_\epsilon)^{-1}. 
\end{align*}
In particular, at $\tau = 1$, we have,
\begin{align*}
    \Lambda_\epsilon^\prime(1) & = 1 - \left(1 - \frac{1}{\delta} \right) \cdot \frac{(1+\epsilon)^2 + 1}{(1+\epsilon)^2}\\ & \implies \lim_{\epsilon \downarrow 0} \Lambda^\prime(1) = \frac{2-\delta}{\delta},
\end{align*}
and, 
\begin{align*}
    \psi_1^\epsilon(1) & = 2+ \epsilon. 
\end{align*}
We consider the following two cases.

\textit{Case 1: $1<\delta < 2$. } Lemma \ref{solutions_real_case} shows that $\Lambda_\epsilon(\tau)$ is convex on $[1,\infty)$. When $\delta < 2$, $\Lambda^\prime_\epsilon(1) > 0$ for $\epsilon$ small enough, and hence $\Lambda_\epsilon$ is strictly increasing and $\tau_r^\epsilon = 1$. Moreover, in this case, for $\epsilon$ small enough,
\begin{align*}
    \frac{\delta}{\delta - 1} & = 2 + \frac{2-\delta}{\delta - 1} > 2 + \epsilon = \psi_1^\epsilon(1).
\end{align*}
Hence, using \eqref{reg_trimming_result},
\begin{align*}
    \lim_{\epsilon \downarrow 0} \lim_{\substack{m,n \rightarrow \infty \\ m = \delta n}} \frac{|\bm x_\star^\UH \hat{\bm x}_\epsilon|^2}{n}  = 0.
\end{align*}

\textit{Case 2: $\delta > 2$} In this case, for small enough $\epsilon$, $\Lambda^\prime_\epsilon(1) < 0$. Hence the $\tau_r^\epsilon$, the minimizer of the convex function $\Lambda_\epsilon$ occurs in the region $(1,\infty)$. This means it satisfies the optimality condition:
\begin{align*}
    \Lambda^\prime_\epsilon(\tau_r^\epsilon) = 0 & \Leftrightarrow \psi_2(\tau_r^\epsilon) = \frac{\delta}{\delta - 1}.
\end{align*}
Next we claim that, $\forall \tau \in [1,\infty)$,
\begin{align*}
    \psi_1^\epsilon(\tau) > \psi_2^\epsilon(\tau) \Leftrightarrow \E[G_\epsilon(\tau)] \cdot \E[|Z|^2 G_\epsilon(\tau)] > \E[G^2_\epsilon(\tau)],
\end{align*}
which is a consequence of Chebychev's association inequality (Fact \ref{chebychev_association}) with the choice:
\begin{align*}
    B &= G_\epsilon(\tau), \; A = |Z|, \\ \; f(a) &=  a^2 \left(\tau  -\T_\epsilon \left( \frac{a}{\sqrt{\delta}} \right)\right), \; g(a) = \left(\tau -\T_\epsilon \left( \frac{a}{\sqrt{\delta}} \right) \right)^{-1}.
\end{align*}
In particular we have $\psi_1^\epsilon(\tau_r^\epsilon) > \delta/(\delta-1)$, and hence Theorem~\ref{result_lambda1} gives us:
\begin{enumerate}
    \item There exists a unique solution $\theta_\star^\epsilon \in (\tau_r^\epsilon,\infty)$ such that $\psi_1^\epsilon(\theta_\star^\epsilon) = \delta/(\delta - 1)$,
    \item and,
    \begin{align*}
        \frac{|\bm x_\star^\UH \hat{\bm x}_\epsilon|^2}{n} & \explain{a.s.}{\rightarrow}  
	\frac{\left(\frac{\delta}{\delta -1 } \right)^2 - \frac{\delta}{\delta -1 } \cdot \psi_2^{\epsilon}(\theta_\star^\epsilon)}{\psi_3^\epsilon(\theta_\star^\epsilon)^2 - \frac{\delta}{\delta -1 } \cdot \psi_2^\epsilon(\theta_\star^\epsilon)}.
    \end{align*}
\end{enumerate}
Next we claim that,
\begin{align*}
    1 < \liminf_{\epsilon \downarrow 0} \theta_\star^\epsilon \leq \limsup_{\epsilon \downarrow 0} \theta_\star^\epsilon <\infty. 
\end{align*}
To see this, observe 
\begin{align*}
    \psi_1^\epsilon(\theta_\star^\epsilon) & =  \frac{\E \frac{|Z|^2 (|Z|^2 + \epsilon)}{(\theta_\star^\epsilon -1)(|Z|^2 + \epsilon) +1}}{\E \frac{(|Z|^2 + \epsilon)}{(\theta_\star^\epsilon -1)(|Z|^2 + \epsilon) +1}}.
\end{align*}
If $\liminf_{\epsilon \downarrow 0} \theta_\star^\epsilon = 1$, one can select a subsequence along which $\psi_1^\epsilon(\theta_\star^\epsilon)  \rightarrow \E |Z|^4 = 2$ by dominated convergence which contradicts: $\psi_2^\epsilon(\theta_\star^\epsilon) = \delta/(\delta - 1) < 2$. Likewise if $\limsup_{\epsilon \downarrow 0} \theta_\star^\epsilon = \infty$, one can find a subsequence along which $\theta_\star^\epsilon \rightarrow \infty$ and, by dominated convergence,
\begin{align*}
    \psi_1^\epsilon(\theta_\star^\epsilon) & = \frac{\E \frac{|Z|^2 (|Z|^2 + \epsilon)(\theta_\star^\epsilon-1)}{(\theta_\star^\epsilon -1)(|Z|^2 + \epsilon) +1}}{\E \frac{(|Z|^2 + \epsilon)(\theta_\star^\epsilon-1)}{(\theta_\star^\epsilon -1)(|Z|^2 + \epsilon) +1}} \rightarrow 1,
\end{align*}
which contradicts $\psi_1^\epsilon(\theta_\star^\epsilon) = \delta/(\delta-1) < 1 \; \forall \; \delta \; \in \; (2,\infty)$. We can now conclude that,
\begin{align*}
    \liminf_{\epsilon \downarrow 0} \theta_\star^\epsilon = \limsup_{\epsilon \downarrow 0} \theta_\star^\epsilon = \theta_\star^\opt,
\end{align*}
where $\theta_\star^\opt$ is the unique solution to $\psi_1^\opt(\tau) = \delta/(\delta - 1)$ in $\tau \in (1,\infty)$ guaranteed by Proposition \ref{junjie_impossibility_result} (due to \cite{ma2019spectral}). This is because, by selecting a subsequence along with ${\theta_\star^\epsilon \rightarrow \liminf_{\epsilon \downarrow 0} \theta_\star^\epsilon}$, we can conclude that, along that subsequence, 
\begin{align*}
    \frac{\delta}{\delta - 1} & = \psi_1^\epsilon(\theta_\star^\epsilon) \rightarrow \psi_1^\opt \left(\liminf_{\epsilon \downarrow 0} \theta_\star^\epsilon\right).
\end{align*}
This implies,
\begin{align*}
     \psi_1^\opt \left(\liminf_{\epsilon \downarrow 0} \theta_\star^\epsilon\right) = \frac{\delta}{\delta-1},
\end{align*}
and analogously,
\begin{align*}
     \psi_1^\opt \left(\limsup_{\epsilon \downarrow 0} \theta_\star^\epsilon\right) = \frac{\delta}{\delta-1}.
\end{align*}
Since Proposition \ref{junjie_impossibility_result} guarantees that the equation ${\psi_1^\opt(\tau) = \delta/(\delta - 1)}$ has a unique solution in $(1,\infty)$ we get,
\begin{align*}
    \liminf_{\epsilon \downarrow 0} \theta_\star^\epsilon = \limsup_{\epsilon \downarrow 0} \theta_\star^\epsilon = \theta_\star^\opt.
\end{align*}
Dominated convergence now yields,
\begin{align*}
    \psi_i^\epsilon(\theta_\star^\epsilon) \rightarrow \psi_i^\opt(\theta_\star^\opt), \text{ as $\epsilon \downarrow 0$ } \forall \; i = 1,2,3,
\end{align*}
and consequently, almost surely,
\begin{align*}
    \lim_{\epsilon \downarrow 0} \lim_{\substack{m,n \rightarrow \infty, \\ m = n \delta }} \frac{|\bm x_\star^\UH \hat{\bm x}_\epsilon|^2}{n} & \explain{a.s.}{=}  
	\frac{\left(\frac{\delta}{\delta -1 } \right)^2 - \frac{\delta}{\delta -1 } \cdot \psi_2^{\opt}(\theta_\star^\opt)}{\psi_3^\opt(\theta_\star^\opt)^2 - \frac{\delta}{\delta -1 } \cdot \psi_2^\opt(\theta_\star^\opt)}.
\end{align*}
The right hand side of the above display can be simplified to:
\begin{align*}
    \frac{\left(\frac{\delta}{\delta -1 } \right)^2 - \frac{\delta}{\delta -1 } \cdot \psi_2^{\opt}(\theta_\star^\opt)}{\psi_3^\opt(\theta_\star^\opt)^2 - \frac{\delta}{\delta -1 } \cdot \psi_2^\opt(\theta_\star^\opt)} & =\frac{\theta_\star^\opt - 1}{\theta_\star^\opt - \frac{1}{\delta}}.
\end{align*}
This clean formula is due to \cite{ma2019spectral} and we refer the reader to Appendix B in \cite{ma2019spectral} for a proof. 


\section{Miscellaneous results}
\label{analytic_appendix}
\begin{fact}[Chebychev Association Inequality, \cite{boucheron2013concentration}] Let $A,B$ be r.v.s and $B \geq 0$. Suppose $f,g$ are two non-decreasing functions. Then, $$\E[B]\E[B f(A) g(A)] \geq \E[f(A)B]\E[g(A) B].$$
        \label{chebychev_association}
Furthermore, if, $\p{B = 0} = 0$ and, 
\begin{align*}
    \p{f(A) = x} = 0, \; \p{g(A) = x} = 0, \; \forall \; x \; \in \; \mathbb R,
\end{align*}
then, the above inequality is strict. 
\end{fact}
\begin{proof}
The proof of the inequality appears in \cite{boucheron2013concentration}. Inspecting the proof we can derive a sufficient condition for the inequality to be strict. The proof in \cite{boucheron2013concentration} shows,
\begin{align*}
    & 2\cdot(\E[B]\E[B f(A) g(A)] - \E[f(A)B]\E[g(A) B])  = \\ &\hspace{3cm} \E B B^\prime (f(A) - f(A^\prime)) \cdot (g(A) - g(A^\prime)).
\end{align*}
where $(B^\prime,A^\prime)$ is an independent sample of the random variables $(B,A)$. Since, $f,g$ are increasing $(f(A) - f(A^\prime)) \cdot (g(A) - g(A^\prime)) \geq 0$ and $B \geq 0, B^\prime \geq 0$. Hence the equality is tight iff:
\begin{align*}
    B B^\prime (f(A) - f(A^\prime)) \cdot (g(A) - g(A^\prime)) & \explain{a.s.}{=} 0,
\end{align*}
which is ruled out by the assumptions of the claim. 
\end{proof}

\section*{Acknowledgments} 
We would like to thank Professor Serban Belinschi for discussions about free probability and Professor Tomoyuki Obuchi for discussions about the replica method. We acknowledge support from NSF DMS-1810888 and  the Google faculty award.

\ifCLASSOPTIONcaptionsoff
  \newpage
\fi



%
\bibliographystyle{unsrt}
\bibliography{ref}

\end{document}